\documentclass[%
 aip,
 sd,%
amsmath,amssymb,
twocolumn,
10pt,
a4paper,
nofootinbib,
tightenlines
]{revtex4-1}

\usepackage
{
    graphicx,
    amssymb,
    amsmath,
    amsthm,
    dsfont, 
    xcolor,
    enumitem,
    mathtools,
    ifthen,
    array,
    multirow,
    adjustbox,
    bbm,
    mdframed,
    upgreek
}

\usepackage[pdffitwindow=false,
            plainpages=false,
            pdfpagelabels=true,
            pdfpagemode=UseOutlines,
            pdfpagelayout=SinglePage,
            bookmarks=false,
            colorlinks=true,
            hyperfootnotes=false,
            linkcolor=blue,
            urlcolor=blue!30!black,
            citecolor=green!50!black]{hyperref}

\newtheorem{theorem}{Theorem}[section]

\newtheorem{lemma}[theorem]{Lemma}
\newtheorem{proposition}[theorem]{Proposition}
\newtheorem{definition}[theorem]{Definition}
\theoremstyle{definition}
\newtheorem{example}[theorem]{Example}
\newtheorem{remark}[theorem]{Remark}

\newtheorem{textalgorithm}[theorem]{Algorithm}

\newcommand{\ts}{\hspace*{0.1em}} 

\newcommand{\R}{\mathbb{R}}                                     
\newcommand{\innerprod}[2]{\left\langle #1,\, #2 \right\rangle} 
\newcommand{\dd}{\mathrm{d}}                                    
\providecommand{\norm}[1]{\left\lVert #1 \right\rVert}          

\newcommand{\overbar}[1]{\mkern 1.5mu\overline{\mkern-1.5mu#1\mkern-1.5mu}\mkern 1.5mu}
\newcommand\restr[2]{{\left.\kern-\nulldelimiterspace #1 \vphantom{\big|} \right|_{#2}}} 

\newcommand\xqed[1]{\leavevmode\unskip\penalty9999 \hbox{}\nobreak\hfill \quad\hbox{#1}}
\newcommand{\exampleSymbol}{\xqed{$\blacktriangle$}} 

\newcommand{\inspace}{\mathbb{X}} 

\newcommand{\ebd}[1][]{
   \ifthenelse{\equal{#1}{}}{\mathcal{E}}{\mathcal{E}_{#1}}}
\newcommand{\pro}[1][]{
   \ifthenelse{\equal{#1}{}}{\mathcal{Q}}{\mathcal{Q}_{#1}}}

\newcommand{\pf}[1][]{
   \ifthenelse{\equal{#1}{}}{\mathcal{P}}{\mathcal{P}_{#1}}}
\newcommand{\epf}[1][]{
   \ifthenelse{\equal{#1}{}}{\widehat{\mathcal{P}}}{\widehat{\mathcal{P}}_{#1}}}
\newcommand{\ko}[1][]{
   \ifthenelse{\equal{#1}{}}{\mathcal{K}}{\mathcal{K}_{#1}}}
\newcommand{\eko}[1][]{
   \ifthenelse{\equal{#1}{}}{\widehat{\mathcal{K}}}{\widehat{\mathcal{K}}_{#1}}}

\newcommand{\cov}[1][]{\mathcal{C}_\mathit{\scriptscriptstyle #1}} 
\newcommand{\ecov}[1][]{\widehat{\mathcal{C}}_\mathit{\scriptscriptstyle #1}} 
\newcommand{\mcov}[1][]{C_\mathit{\scriptscriptstyle #1}} 
\newcommand{\mecov}[1][]{\widehat{C}_\mathit{\scriptscriptstyle #1}} 

\newcommand{\gram}[1][]{G_\mathit{\scriptscriptstyle #1}} 
\newcommand{\agram}[1][]{\overbar{G}_\mathit{\scriptscriptstyle #1}} 
\newcommand{\tagram}[1][]{\widetilde{G}_\mathit{\scriptscriptstyle #1}} 
\newcommand{\rgram}[1][]{\widehat{G}_\mathit{\scriptscriptstyle #1}} 

\DeclareMathOperator{\diag}{diag}

\mathtoolsset{centercolon}

\allowdisplaybreaks

\definecolor{boxback}{gray}{0.95}

\makeatletter
\def\blfootnote{\gdef\@thefnmark{}\@footnotetext}
\makeatother

\begin{document}

\title{A kernel-based approach to molecular conformation analysis}

\author{Stefan Klus}%
\email[]{stefan.klus@fu-berlin.de}
\affiliation{Department of Mathematics and Computer Science, Freie Universit\"at Berlin, 14195 Berlin, Germany}

\author{Andreas Bittracher}
\email[]{bittracher@mi.fu-berlin.de}
\affiliation{Department of Mathematics and Computer Science, Freie Universit\"at Berlin, 14195 Berlin, Germany}

\author{Ingmar Schuster}
\email[]{ingmar.schuster@zalando.de}
\affiliation{Zalando Research, 11501 Berlin, Germany}

\author{Christof Sch\"utte}
\email[]{christof.schuette@fu-berlin.de}
\affiliation{Department of Mathematics and Computer Science, Freie Universit\"at Berlin, 14195 Berlin, Germany}
\affiliation{Zuse Institute Berlin, 14195 Berlin, Germany}%

\blfootnote{S. Klus and A. Bittracher contributed equally to this work.}

\begin{abstract}
We present a novel machine learning approach to understanding conformation dynamics of biomolecules. The approach combines kernel-based techniques that are popular in the machine learning community with transfer operator theory for analyzing dynamical systems in order to identify conformation dynamics based on molecular dynamics simulation data. We show that many of the prominent methods like Markov State Models, EDMD, and TICA can be regarded as special cases of this approach and that new efficient algorithms can be constructed based on this derivation. The results of these new powerful methods will be illustrated with several examples, in particular the alanine dipeptide and the protein NTL9.
\end{abstract}

\maketitle

\section{Introduction}

The spectral analysis of transfer operators such as the Perron--Frobenius or Koopman operator is by now a well-established technique in molecular conformation analysis \cite{Schuette1999, pande2010everything, msm_milestoning, CN14}. The goal is to estimate---typically from long molecular dynamics trajectories---the slow conformational changes of molecules. These slow transitions are critical for a better understanding of the functioning of peptides and proteins. Transfer operators propagate probability densities or observables. Since these operators are infinite-dimensional, they are typically projected onto a space spanned by a set of predefined basis functions. The integrals required for this projection can be estimated from training data, resulting in methods such as \emph{extended dynamic mode decomposition (EDMD)}~\cite{WKR15, KKS16} or the \emph{variational approach of conformation dynamics (VAC)}~\cite{NoNu13, NKPMN14}. For a comparison of these methods, see Ref.~\onlinecite{KNKWKSN18}. A similar framework has been introduced by the machine learning community. Probability densities as well as conditional probability densities can be embedded into a \emph{reproducing kernel Hilbert space} (RKHS)~\cite{MFSS16}. These embedded probability densities are then mapped forward by a conditional mean embedding operator, which can be viewed as an embedded Perron--Frobenius operator~\cite{KSM17}. An approximation of this operator in the RKHS can be estimated from training data, leading to methods that are closely related to EDMD and its variants. In the same way, this approach can be extended to the conventional Perron--Frobenius and Koopman operator, under the assumption that the probability densities or observables are functions in the RKHS \cite{KSM17}.

In this paper, we will derive different kernel-based empirical estimates of transfer operators. There are basically two ways to obtain a finite-dimensional approximation of such operators: The first one explicitly uses the feature map representation of the kernel, the second is based on kernel evaluations only. The advantage of the former is that the size of the resulting eigenvalue problem depends only on the size of the feature space, but not on the size of the training data set (this corresponds to EDMD or VAC). However, this approach can in general not be applied to the typically high-dimensional systems prevalent in molecular dynamics due to the curse of dimensionality and furthermore requires an explicit feature space representation, i.e., an explicit basis of the approximation space. For the kernel-based variant, the size of the eigenvalue problem is independent of the number of basis functions---and thus allows for implicitly infinite-dimensional feature spaces---, but depends on the size of the training data set (this corresponds to kernel EDMD or kernel TICA). Kernel-based methods thus promise increased performance and accuracy in transfer operator-based conformation analysis. The framework presented within this work allows us to derive both approaches (classical and kernel-based) as different approximations of kernel transfer operators and thus subsumes the aforementioned methods.

We will show that the eigenvalue problem for a given transfer operator can be transformed by rewriting eigenfunctions in terms of the integral operator associated with the kernel. The collocation-based discretization of this modified problem formulation then directly results in purely kernel-based methods, without having to \emph{kernelize} existing methods. (By \emph{kernelize} we mean formally rewriting an algorithm in terms of kernel evaluations only.) Furthermore, we show that for stochastic systems it is possible to circumvent the drawback that the amount of usable training data is limited by the maximum size of the eigenvalue problem that can be solved numerically using averaging techniques, which can result in more accurate and smoother eigenfunctions. We will apply the derived algorithms to realistic molecular dynamics data and analyze the alanine dipeptide and the protein NTL9. The conformations obtained for NTL9 are comparable to results computed using deep neural networks, see Ref.~\onlinecite{Noe_VAMPnets2018}, and demonstrate that our approach is competitive with state-of-the-art deep learning methods.

The remainder of this paper is structured as follows: In Section~\ref{sec:Prerequisites}, we will briefly introduce transfer operators, reproducing kernel Hilbert spaces, and the integral operator associated with the kernel. We will then show  in Section~\ref{sec:Kernel-based discretization of eigenvalue problems} that methods to estimate eigenfunctions of transfer operators from training data can be derived using the integral operator associated with the kernel. Section~\ref{sec:Applications} illustrates that it is possible to obtain accurate eigenfunction approximations even for high-dimensional molecular dynamics problems. Future work and open questions are discussed in Section~\ref{sec:Conclusion}. The appendix contains detailed derivations of the methods, which are based on the Mercer feature space representation of a kernel, and proofs. Moreover, we will define kernel transfer operators and compare our proposed algorithms with other well-known methods such as kernel EDMD and kernel TICA, which have been proposed in Refs.~\onlinecite{WRK15, SP15}.

\section{Prerequisites}
\label{sec:Prerequisites}

In what follows, we will introduce transfer operators, kernels, and reproducing kernel Hilbert spaces. The notation used throughout the manuscript is summarized in Table~\ref{tab:Notation}.

\begin{table}[tb]
    \centering
    \caption{Overview of notation.}
    \begin{tabular}{ll}
        \hline
        $ m $                              & number of snapshots \\
        $ n $                              & number of basis functions \\
        $ k $                              & kernel \\
        $ \mathcal{E}_k $                  & integral operator associated with $ k $ \\
        $ \gamma_\imath, \phi_\imath $     & eigenvalues and eigenfunctions of the integral operator \\
        $ \ko, \pf $                       & Koopman and Perron--Frobenius operator \\
        $ \ko[k], \pf[k] $                 & kernel Koopman and Perron--Frobenius operator \\
        $ \lambda_\imath, \varphi_\imath $ & eigenvalues and eigenfunctions of transfer operators \\
        \hline
    \end{tabular}
    \label{tab:Notation}
\end{table}

\subsection{Transfer operators}
\label{sec:transferoperators}

We briefly motivate the usefulness of transfer operators for conformational metastability analysis. For a more detailed introduction, we refer to Refs.~\onlinecite{LaMa94, KKS16}.

Let the molecular process be given by a stochastic process $X_t$ on some high-dimensional state space $ \inspace \subseteq \R^d $, for instance, the space of Cartesian coordinates of all atoms of the system.
One can imagine the system as a random walk in some high-dimension potential energy landscape, i.e., $X_t$ could be described by an overdamped Langevin equation (see Example \ref{ex:4well_MSMexample}), but we do not require this explicitly.

As a stochastic process, the pointwise evolution of $X_t$ is formally described by the system's \emph{transition density function} $ p_\tau(y \mid x) $, which gives the probability to find the system at a point $y$ after some \emph{lag time} $\tau$, given that it started in $x$ at time 0. 
Using $p_\tau$, we can express the evolution of arbitrary observables and densities under the dynamics using so-called \emph{transfer operators}.
Consider some observable $f \colon \inspace\to\mathbb{R}$ of the system, and the value $f(x)$ of this observable at some point $x\in\inspace$. Then the expected value of this observable of the system that started in $x$ and evolved to time $\tau$ is given by
\begin{equation}
    f_\tau(x) := \ko f(x) := \int p_\tau(y \mid x) \ts f(y) \ts \dd \mu(y).
\end{equation}
The operator $\ko$ is called the \emph{Koopman operator} of the system. In addition to the lag time $\tau$, it depends on the initial measure $\mu$, which in our case will be either the Lebesgue measure or the system's invariant measure (also known as Gibbs measure). This will lead to slightly different interpretations of the evolved observables (more on that later) and allow us to utilize data points sampled in different ways in the algorithms.

Similarly, we can describe the evolution of densities: Assume the initial distribution of the system is given by the density $u \colon \inspace\to \mathbb{R}$. Then the distribution of the evolved system at time $\tau$ is given by
\begin{equation*}
    u_\tau(x) := \pf u(x) := \int p_\tau(x \mid y) \ts u(y) \ts \dd \mu(y),
\end{equation*}
where $\pf$ is called the \emph{Perron--Frobenius operator} of the system.
In the special case of $X_\tau$ being a deterministic system, i.e., $X_{t+\tau} = \Theta(X_t)$ with some invertible flow map $\Theta\colon\inspace\rightarrow\inspace$, we get the simplified definitions
\begin{equation*}
    \ko f(x) = f\big(\Theta (x)\big) 
    \quad \text{and} \quad
    \pf u(x) = u\big(\Theta^{-1} (x)\big).
\end{equation*}


Even though describing the system's dynamics by operators may seem overly abstract compared to for example an SDE or the associated transport equations, this formalism has proven useful for metastability analysis.
In fact, the basis of a multitude of computational methods is the realization that all information about the long-term behavior of the dynamics is contained within certain eigenfunctions of the transfer operators.
Consider for instance the eigenproblem of the Koopman operator
\begin{equation}
\label{eq:operator KO eigenproblem}
\ko \varphi = \lambda \varphi,
\end{equation}
where the solution $\varphi$ is a function in a certain Hilbert space.
Under mild conditions, all eigenvalues $\lambda$ are real and bounded from above by 1. Crucially, eigenfunctions corresponding to eigenvalues \emph{close to 1} reveal the location of metastable sets in $\mathbb{X}$ \cite{SS13,DJ99,KNKWKSN18}. Analogue results hold for the Perron Frobenius eigenvalue problem
\begin{equation}
\label{eq:operator PF eigenproblem}
\pf \varphi = \lambda \varphi.
\end{equation}

\begin{example} \label{ex:4well_MSMexample}
We consider the process $X_t$ described by the overdamped Langevin equation, i.e., the stochastic differential equation (SDE)
\begin{equation*}
    \mathrm{d}\mathbf{X}_t = -\nabla V(\mathbf{X}_{t})\,\mathrm{d}t + \sqrt{2\beta^{-1}} \, \mathrm{d}\mathbf{W}_t
\end{equation*}
with the quadruple-well potential $ V(x) = (x_1^2 - 1)^2 + (x_2^2 - 1)^2 $ and the inverse temperature $ \beta = 4 $. Here, $\mathbf{W}_t$ describes a standard Wiener process (white noise). The potential on the domain $ \mathbb{X} = [-2, 2] \times [-2, 2] $ is illustrated in Figure~\ref{fig:4well_MSMexample} (a).

As demonstrated in Figure~\ref{fig:4well_MSMexample} (b)--(d), the location of the metastable sets in $\inspace$ can be identified by the sign structure of the dominant eigenfunctions. This can be done algorithmically by applying a \emph{spectral clustering} algorithm such as PCCA+ to the eigenfunctions. The end goal is the construction of a Markov State Model (MSM), in which the main metastable states form the states of a Markov chain\cite{SS13}.
In this example, this model would be entirely described by a $4\times 4$ transition rate matrix. This is a vastly reduced model compared to the original two-dimensinal SDE, and still preserve the statistics of the long-time metastable transitions.

This system will be revisited in Section~\ref{sec:Applications}, where it will be analyzed with the newly derived kernel-based methods.
Moreover, the application of PCCA+ to approximated eigenfunctions will be demonstrated in Sections~\ref{sec:AlanineDipeptide} and~\ref{sec:NTL9}.

\begin{figure*}[t]
    \centering
    \includegraphics[width=\textwidth]{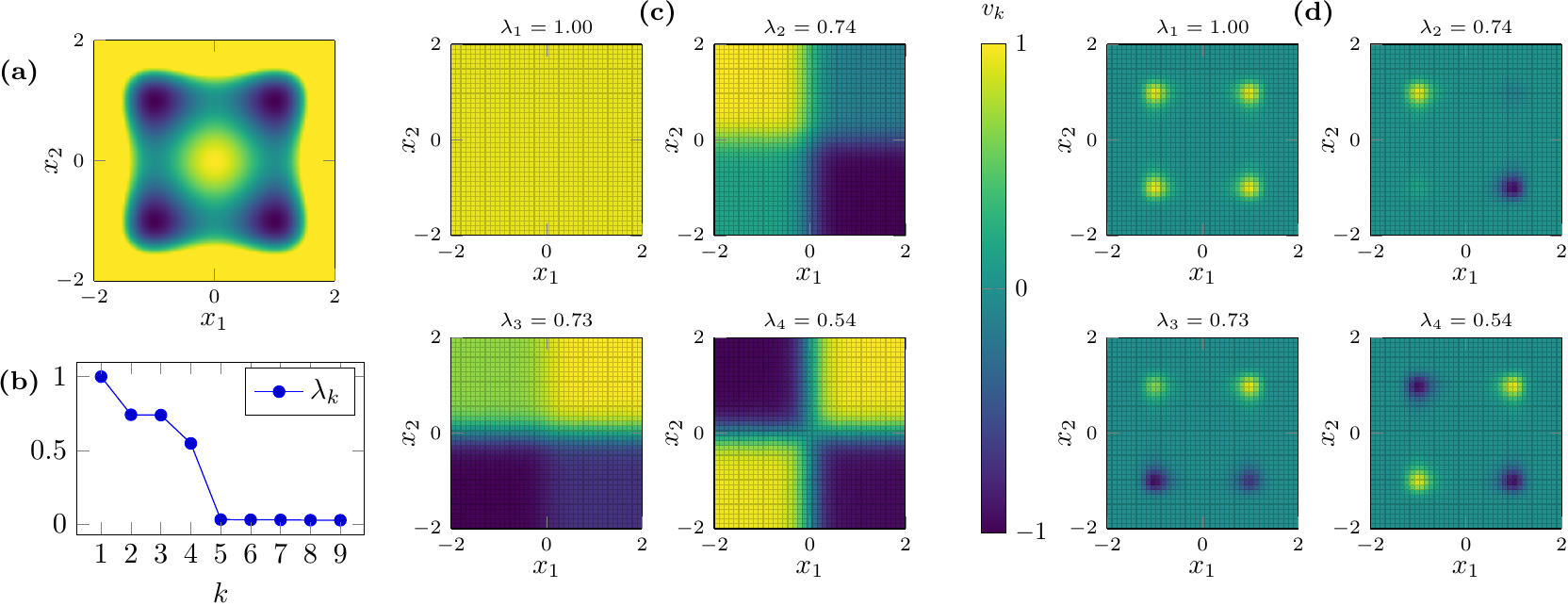}
    \caption{Transfer operator-based metastability analysis. (a) Potential $V$ with four local minima. (b) Largest eigenvalues of the discretized Koopman operator, independent of the choice of the measure $\mu$. A spectral gap after $\lambda_4$ indicates that the dominant four eigenfunctions are relevant to metastability analysis. (c) Dominant eigenfunctions of $\ko$ with $\mu$ chosen as the Lebesgue measure. We observe that the eigenfunctions are almost constant in the four quadrants. This indicates that each of the quadrants is metastable, i.e., trajectories that start in each quadrant typically stay in it for a long time. (d) Dominant eigenfunctions of $\ko$ with $\mu$ chosen as the Gibbs measure. Their sign structure reveals the ``cores'' of the metastable sets. Formally, they differ from the Lebesgue eigenfunctions only by a weighting with the invariant density.}
    \label{fig:4well_MSMexample}
\end{figure*}

\exampleSymbol
\end{example}

\subsubsection*{Galerkin approximation of transfer operators}

Formally, the operators $\ko$ and $\pf$ are defined on the function space $L_2(\mu)$, which is an infinite-dimensional Hilbert space with inner product $\langle \cdot , \cdot\rangle_2$.
The solution to the above eigenproblems are therefore again functions in $L_2(\mu)$. 
Thus, in order to be numerically solvable, the problems need to be \emph{discretized}, i.e., approximated by finite, vector-valued eigenproblems. We demonstrate this with the aid of the Koopman eigenproblem \eqref{eq:operator KO eigenproblem}. For \eqref{eq:operator PF eigenproblem}, the procedure works in a similar fashion.

Consider a set of $n$ linearly independent functions $\{v_1,\dots,v_n\}\subset L_2(\mu)$, and all the functions that can be expressed by linear combinations of the these functions:
\begin{equation*}
    \widetilde{\varphi}(x) = \sum_{\imath=1}^n c_\imath \ts v_\imath(x),\quad c \in \R^n.
\end{equation*}
Each such function is uniquely defined by its \emph{coefficient vector} $c$. Together, these functions form an $n$-dimensional subspace of $L_2(\mu)$, called the \emph{Galerkin space}, and denoted by $\mathcal{V}$ in what follows.

If $\mathcal{V}$ is a ``rich'' approximation space---particularly, if its dimension $n$ is high enough---it can reasonably be assumed that the eigenproblem \eqref{eq:operator KO eigenproblem} can be \emph{approximately} solved by an element $\widetilde{\varphi}$ from~$\mathcal{V}$:
\begin{equation} \label{eq:operator KO eigenproblem approx}
    \ko \widetilde{\varphi} \approx \lambda \widetilde{\varphi}.
\end{equation}
There now exist various ways to construct from \eqref{eq:operator KO eigenproblem approx} a well-defined vector--valued eigenproblem of the form
\begin{equation} \label{eq:operator KO eigenproblem discrete}
    Tc = \lambda c
\end{equation}
for the coefficient vector $c$ of $\widetilde{\varphi}$. Conceptually, however, the most important step for discretizing the eigenproblem has already been completed, namely by assuming that its (approximate) solution lies in the finite-dimensional function space $\mathcal{V}$. Thus, the greatest influence on the final approximation error is the choice of the finite basis $\{v_1,\ldots,v_n\}$.

For example, dividing $\inspace$ regularly into $n$ boxes, $\inspace = A_1\cup\ldots\cup A_n$, the characteristic functions
\begin{equation*}
\mathbbm{1}_{A_\imath}(x) = \begin{cases}
	1, & \text{if } x\in A_\imath, \\
	0, & \text{otherwise},
	\end{cases}
\end{equation*}
for $\imath=1,\dots,n$ can be used as Galerkin basis. In this case the matrix $T$ in~\eqref{eq:operator KO eigenproblem discrete} is the $n \times n$ \emph{transition matrix} between the boxes. Its elements $T_{\imath\jmath}$ describe the probability of the system to end up at in box $A_\jmath$ after starting in $A_\imath$, which can be estimated from simulation data.

Other commonly-used types of Galerkin basis functions include polynomials or trigonometric polynomials over $\inspace$ up to a certain degree. 

The size of the discretized eigenvalue problem equals the number of basis functions, determined for instance by the resolution of the box partition or the maximum degree of the polynomials. In particular, this approach suffers from the \emph{curse of dimensionality}, the phenomenon that for ansatz functions $v_\imath$ chosen on a regular grid, the effort grows exponentially with the dimension $d$. This renders the Galerkin approach prohibitively expensive for high-dimensional problems. 
We will see below how kernel methods---by \emph{implicitly} defining bases for the solution space---can mitigate these problems.

\begin{remark}
With the Galerkin approach the derivation of a vector-valued eigenvalue problem~\eqref{eq:operator KO eigenproblem discrete} is quite straightforward once an explicit basis of the Galerkin space is given. However, in order for the kernel-based approach with its implicit basis to show its advantages, we need to keep working with the original operator formulation of the eigenproblems \eqref{eq:operator KO eigenproblem} and \eqref{eq:operator PF eigenproblem} for now. This will also enable us to systematically separate approximation errors that arise from restricting the solution space from errors that are due to finite sampling of the dynamics. In the end, however, we will again obtain a data-based matrix approximation of the original eigenproblem.
\end{remark}

\subsection{Kernels and their approximation spaces}

Moving away from dynamical systems and their transfer operators for a moment, we introduce kernels and their reproducing kernel Hilbert spaces. This introduction is based on Refs.~\onlinecite{Schoe01, StCh08}.

Let $k \colon \inspace \times \inspace \to \mathbb{R}$ be a positive definite function, called the \emph{kernel}. Associated with each kernel is a unique space of functions $\mathbb{H}$, called the \emph{reproducing kernel Hilbert space (RKHS)}, that will serve as an analogue to the previously-introduced Galerkin space $\mathcal{V}$. However, the elements of $\mathbb{H}$ in general cannot be written by a finite combination of pre-defined basis functions. Instead, the defining property of $\mathbb{H}$ is that each $f\in\mathbb{H}$ can be linearly combined by evaluations of the kernel:
$$
f(x') = \sum_{\imath\in I} c_\imath k(x', x_\imath).
$$
It is important to note that the set of evaluation points $\{x_\imath\}_{\imath\in I}$, as well as the number of elements of the index set~$I$, depends on the function $f$. Thus, for each element of $\mathbb{H}$, an uncountable infinite number of ``basis functions'' $\{k(x,\cdot)\}_{x\in\inspace}$ are available to express it. It is therefore no surprise that the approximation space $\mathbb{H}$ can be incomparably ``richer'' than the Galerkin approximation space. However, how rich the space ultimately is depends on the choice of the kernel.

\begin{example}
\label{ex:kernels}
Examples of commonly-used kernels are:
\begin{enumerate}[label=(\roman*), wide, labelwidth=!, labelindent=0pt]
\item Polynomial kernel of degree $ p $: 
$$
k(x,x') = (x^\top x' + c)^p
$$ 
with $c>0$. It can be shown that in this case $\mathbb{H}$ is identical to the Galerkin space of polynomials over $\inspace$ up to degree~$p$.
\item Gaussian kernel: 
$$
k(x,x') = \exp\left(-\frac{1}{\sigma}\norm{x-x'}_2^2\right),
$$
where $ \sigma > 0 $ is a parameter. This important kernel induces an infinite-dimensional RKHS $\mathbb{H}$ that is dense in $L_2(\mu)$\cite{StCh08}, i.e., a very rich approximation space.
\item Given functions $ f_\imath \colon \inspace \to \R $, with $ \imath = 1, \dots, n $, then
$$ 
k(x,x') = \sum\limits_{\imath=1}^n f_\imath(x) f_\imath(x') 
$$
defines a kernel and $\mathbb{H}$ is the Galerkin space spanned by the basis $\{f_1,\dots,f_n\}$ \exampleSymbol
\end{enumerate}
\end{example}

A kernel can be interpreted as a similarity measure. In point (iii) in the example above, this similarity measure is just the standard Euclidean inner product between the $n$-dimensional \emph{features} $f(x)$ and $f(x')$ of $x$ and $x'$.

However, the kernel is typically not defined by a collection of explicitly given features. Rather, each kernel implicitly \emph{defines} a possibly infinite collection of features. Thus, by choosing the kernel as some similarity measure that is ``relevant'' to the molecular system at hand, one can formulate and solve the eigenvalue problem in the possibly infinite-dimensional feature space $\mathbb{H}$ with highly expressive features, \emph{without} having to explicitly construct, compute, and discretize those features.

For example, the kernel $k(x,x')$ may be based on comparing a large collection of unrelated but independently relevant observables at the points $x$ and $x'$, drawn from a database, or measure an appropriate distance between possibly high-dimensional or unwieldy representations of $x$ and $x'$, for which however distances can more easily be computed.
In Section~\ref{sec:NTL9}, we will see how to construct a kernel from a distance measure based on the contact map representation of proteins.
Moreover, as kernels do not necessarily require vectorial input data, a kernel based on the graph representation of a molecule can also be defined\cite{Borgwardt2005,Vishwanathan2010}.

\subsubsection*{Integral operators}

In order to formulate eigenproblems in $\mathbb{H}$ instead of the original space $L_2(\mu)$, we require the following transformation map:

\begin{equation*}
    \mathcal{S}_k f(x) = \int k(x, x^\prime) \ts f(x^\prime) \ts \dd \mu(x^\prime).
\end{equation*}

The operator $S_k \colon L_2(\mu)\rightarrow \mathbb{H}$ is called the \emph{integral operator}. Although it looks like a transfer operator, these concepts are unrelated---there is no dynamics in the definition of $ \mathcal{S}_k $. Applying this integral operator to a function can be viewed as a transformation, similar to a Fourier transform. It is the infinite-dimensional equivalent of a basis change.

For technical reasons, it is advantageous to think of the image space of $ \mathcal{S}_k $ not as $\mathbb{H}$ but again as $L_2(\mu)$. We therefore additionally define the \emph{embedding operator} $\ebd[k] \colon L_2(\mu) \to L_2(\mu)$
$$
	\mathcal{E}_kf = \operatorname{id}\big(\mathcal{S}_kf\big),
$$
where $\operatorname{id}\colon \mathbb{H} \to L_2(\mu)$ is the identification of a function with its equivalence class in $L_2(\mu)$. The relationship is illustrated in the following diagram:
\begin{minipage}{\linewidth}
    \centering
    \vspace*{1ex}
    \includegraphics[width=0.5\linewidth]{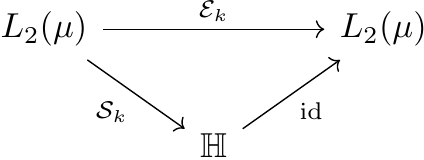} \\
    \vspace*{1ex}
\end{minipage}
Moreover, every function $ f \in \mathbb{H} $ is $ L_2(\mu) $-integrable and, in what follows, we will consider $ \mathbb{H} $ as a subset of $ L_2(\mu) $, although this is technically not entirely correct since $ \mathbb{H} $ contains functions and $ L_2(\mu) $ equivalence classes.

Now, instead of discretizing the transfer operator eigenvalue problem directly, we will first rewrite it in terms of \emph{transformed} eigenfunctions. A discretization of this reformulated eigenvalue problem will finally result in algorithms that require only evaluations of the kernel in given data points.

\section{Kernel-based discretization of eigenvalue problems}
\label{sec:Kernel-based discretization of eigenvalue problems}

In this section, we derive algorithms for the numerical solution of transfer operator eigenvalue problems in reproducing kernel Hilbert spaces that rely only on kernel evaluations in sampling data. Summarized, the main advantages of this approach are:
\begin{itemize}[wide, labelwidth=!, labelindent=0pt]
\item Depending on the chosen kernel, the approximation space $\mathbb{H}$ can in principle be infinite-dimensional and dense in $L_2(\mu)$, allowing for accurate approximations of transfer operators and their eigenfunctions.

\item Using a standard Galerkin approximation, the size of the eigenvalue problem depends on the number of basis functions and thus grows rapidly with increasing dimension $ d $ if we want to maintain a certain approximation error. The size of the eigenvalue problem associated with kernel-based methods, on the other hand, depends only on the number of data points, independently of the dimension $ d $. For increasing $ d $, however, we then have to take into account more data points to attain a certain approximation error.

\item Instead of explicitly constructing (high-dimensional) feature spaces, kernels based on chemically relevant distance measures can be used that implicitly work in these spaces.
\end{itemize}

The derivations of the new algorithms are surprisingly simple once we assume the following:
\begin{enumerate}[wide, labelwidth=!, labelindent=0pt]
\item The eigenfunctions $ \varphi $ are in $ \mathbb{H} $. That is, we consider transfer operators restricted to $ \mathbb{H} $. This is similar to the Galerkin approach, where we seek the solution to the eigenproblems only in a finite-dimensional Galerkin space $\mathcal{V}$, with the difference that the RKHS $\mathbb{H}$ can still be infinite-dimensional. Depending on the kernel, this can be a strong restriction, e.g., if $ k $ is a polynomial kernel of low order (whose RKHS is of low finite dimension), or a mild restriction, e.g., if $ k $ is the Gaussian kernel since the resulting RKHS is then infinite-dimensional and dense in $ L_2(\mu) $.
\item For any eigenfunction $ \varphi \in \mathbb{H} $, we can find a function $ \widetilde{\varphi} \in \mathbb{H} $ such that $ \varphi = \mathcal{E}_k \widetilde{\varphi} $. That is, $ \varphi $ is the embedding of $ \widetilde{\varphi} $ into the RKHS. This technical requirement can in fact be shown to always be fulfilled under mild conditions, see Lemma~\ref{lem:operator properties} in the appendix.
\end{enumerate}

\subsection{Koopman operator for deterministic systems}
\label{sec:kEDMD_KO}

We initially consider deterministic dynamical systems, i.e., $ \ko f(x) = f(\Theta(x)) $. In order to obtain kernel-based methods, we will first rewrite the eigenvalue problem \eqref{eq:operator KO eigenproblem} based on the above assumptions. We now want to find functions $ \widetilde{\varphi} \in \mathbb{H} $ such that
\begin{equation}
	\label{eq:Koopman EVP in H}
    \ko \ts \mathcal{E}_k \ts \widetilde{\varphi} = \lambda \ts \mathcal{E}_k \ts \widetilde{\varphi}
\end{equation}
and then set $ \varphi = \mathcal{E}_k \ts \widetilde{\varphi} $ to obtain an eigenfunction. Using the definition of the Koopman operator for deterministic systems we get for the left hand side
\begin{align*}
	\ko \mathcal{E}_k \ts \widetilde{\varphi} (x) &= \ko \int k(x, x^\prime) \ts \widetilde{\varphi}(x^\prime) \ts \dd \mu(x^\prime) \\
	&= \int k(\Theta(x), x^\prime) \ts \widetilde{\varphi}(x^\prime) \ts \dd \mu(x^\prime)~.
\end{align*}
With this the eigenproblem~\eqref{eq:Koopman EVP in H} becomes
\begin{equation} \label{eq:modified Koopman EVP}
       \int k(\Theta(x), x^\prime) \ts \widetilde{\varphi}(x^\prime) \ts \dd \mu(x^\prime) = \lambda \int k(x, x^\prime) \ts \widetilde{\varphi}(x^\prime) \ts \dd \mu(x^\prime),
\end{equation}
which has to be satisfied for all points $x\in\inspace$.

We assume now that we have training data $ (x_i, \, y_i) $, $ i = 1, \, \dots, \, m $, where the $x_i$ are $\mu$-distributed and $ y_i = \Theta(x_i) $. Typically, the data is given either by many short simulations or few long simulations of the system. 
In a first approximation step, we now relax the transformed eigenproblem~\eqref{eq:modified Koopman EVP}, by requiring that it only needs to be satisfied in the data points:
\begin{equation*}
    \int k(y_i, x^\prime) \ts \widetilde{\varphi}(x^\prime) \ts \dd \mu(x^\prime)
        = \lambda \int k(x_i, x^\prime) \ts \widetilde{\varphi}(x^\prime) \ts \dd \mu(x^\prime),
\end{equation*}
for $i=1,\ldots,m$. Finally, in a second approximation step, we estimate the integrals also from the training data by Monte Carlo quadrature, i.e.,
\begin{equation*}
    \sum_{j=1}^m k(y_i, x_j) \ts \widetilde{\varphi}(x_j) = \lambda \sum_{j=1}^m k(x_i, x_j) \ts \widetilde{\varphi}(x_j),
\end{equation*}
for $i=1,\ldots,m$. With the Gram matrix $ \gram[XX] $ and time-lagged Gram matrix $ \gram[YX] $ defined by $ [\gram[XX]]_{ij} = k(x_i, x_j) $ and $ [\gram[YX]]_{ij} = k(y_i, x_j) $, this can be written in matrix form as
\begin{equation*}
    \gram[YX] \ts \widetilde{\varphi}_X = \lambda \ts \gram[XX] \ts \widetilde{\varphi}_X,
\end{equation*}
where $ \widetilde{\varphi}_X = [\widetilde{\varphi}(x_1), \, \dots, \, \widetilde{\varphi}(x_m)]^\top $. This matrix eigenvalue problem can now be solved numerically.

Discretizing a function-valued equation by requiring it to hold only at specific evaluation points is called a \emph{collocation method} in the numerical analysis of differential equations.
To again obtain a solution $\varphi$ that is defined in all of $\inspace$, we also replace the integral $ \varphi = \mathcal{E}_k \widetilde{\varphi} $ by its estimate
\begin{equation*}
    \varphi = \sum_{i=1}^m k(\cdot, x_i) \ts \widetilde{\varphi}(x_i) = \Phi \ts \widetilde{\varphi}.
\end{equation*}
Here, $ \Phi = [k(\cdot, x_1), \, \dots, \, k(\cdot, x_m)] $ can be seen as a row-vector of functions. In the machine learning community, $ \Phi $ is also called feature matrix, although it is technically not a matrix.

\begin{mdframed}[backgroundcolor=boxback,hidealllines=true]
\begin{textalgorithm}
\label{alg:KO_deterministic}
The proposed algorithm for approximating eigenfunctions of the Koopman operator can be summarized as:
\begin{enumerate}
\item Select a kernel $ k $ and compute the Gram matrices $ \gram[XX] $ and $ \gram[YX] $.
\item Solve the eigenvalue problem $ \gram[XX]^{-1} \ts \gram[YX] \ts \widetilde{\varphi} = \lambda \ts \widetilde{\varphi} $.
\item Set $ \varphi = \Phi \ts \widetilde{\varphi} $ as the approximation to the original eigenproblem \eqref{eq:operator KO eigenproblem}.
\end{enumerate}
\vspace{1pt}
\end{textalgorithm}
\end{mdframed}

If one is only interested in $\varphi$ evaluated at the data points, i.e., $ \varphi_X = [\varphi(x_1), \dots, \varphi(x_m)]^\top $, it can be obtained by $ \varphi_X = \gram[XX] \ts \widetilde{\varphi} $ since $ \Phi(x_i) = [k(x_i, x_1), \dots, k(x_i, x_m)] $ is the $ i $th row of $ \gram[XX] $.

\begin{remark}
The matrix $ \gram[XX] $ might be singular or close to singular so that the inverse $ \gram[XX]^{-1} $ does not exist or is ill-conditioned. (For the Gaussian kernel, $ \gram[XX] $ is guaranteed to be regular provided that all points $ x_1, \dots, x_m $, are distinct, but it might still be ill-conditioned.) Thus, for reasons of numerical stability, the above eigenproblem is often replaced by its regularized version
\begin{equation} \label{eq:regularized eigenproblem}
    (\gram[XX] + \eta \ts I) ^{-1} \ts \gram[YX] \ts \widetilde{\varphi} = \lambda \ts \widetilde{\varphi}.
\end{equation}
This corresponds to Tikhonov regularization and reduces overfitting, see Ref.~\onlinecite{Schoe01}.
\end{remark}

Algorithm~\ref{alg:KO_deterministic} is identical to the kernel EDMD formulation proposed in Ref.~\onlinecite{WRK15}, derived for explicitly given finite-dimensional feature spaces by \emph{kernelizing} conventional EDMD. In contrast, we directly derived a solution to the eigenproblem in the RKHS $\mathbb{H}$ and its approximation from data. Additionally, a derivation based on empirical estimates of kernel covariance and cross-covariance operators is described in Ref.~\onlinecite{KSM17}. More details about the feature space and relationships with the aforementioned methods can be found in Appendix~\ref{app:Relationships with other methods}.

\subsection{Perron--Frobenius operator for deterministic systems}
\label{sec:kEDMD_PF}

A collocation-based approximation of the Perron--Frobenius eigenproblem \eqref{eq:operator PF eigenproblem} for deterministic systems can be derived in a similar fashion. It can be shown that
\begin{equation} \label{eq:P to K}
    \mathcal{E}_k \pf f(x) = \int k(x, \Theta(x^\prime)) \ts f(x^\prime) \ts d\mu(x^\prime).
\end{equation}
We again set $ \varphi = \mathcal{E}_k \ts \widetilde{\varphi} $, which results in $ \pf \mathcal{E}_k \widetilde{\varphi} = \lambda \ts \mathcal{E}_k \widetilde{\varphi} $. In order for this equation to be satisfied, the expression on the left must be a function in $ \mathbb{H} $ since we assume $ \varphi $ to be in~$ \mathbb{H} $. Applying $ \mathcal{E}_k $ to both sides leads to transformed functions in $ \mathbb{H} $, we are basically just changing the coefficients of the series expansion, see Lemma~\ref{lem:operator properties}. This results in the following transformed, yet equivalent eigenvalue problem: We seek functions $ \widetilde{\varphi} \in \mathbb{H} $ such that
\begin{equation*}
    \mathcal{E}_k \pf \mathcal{E}_k \widetilde{\varphi} = \lambda \ts \mathcal{E}_k \mathcal{E}_k \widetilde{\varphi}
\end{equation*}
and set $ \varphi = \mathcal{E}_k \widetilde{\varphi} $. Exploiting \eqref{eq:P to K}, where $ f = \mathcal{E}_k \ts \widetilde{\varphi} $, we can rewrite this as
\begin{align*}
    \ebd[k] \ts \pf \ts \mathcal{E}_k \ts \widetilde{\varphi}
        &= \int k(\cdot, y^\prime) \int k(x^\prime, x^{\prime \prime}) \ts \widetilde{\varphi}(x^{\prime \prime}) \ts \dd \mu(x^{\prime \prime}) \ts \dd \mu(x^\prime) \\
        &\overset{!}{=} \lambda \int k(\cdot, x^\prime) \int k(x^\prime, x^{\prime \prime}) \ts \widetilde{\varphi}(x^{\prime \prime}) \ts \dd \mu(x^{\prime \prime}) \ts \dd \mu(x^\prime),
\end{align*}
which we discretize in the same way as above---collocation in $ x_i $, $ i = 1, \, \dots, \, m $, and Monte Carlo approximation of the integrals---and obtain the discrete eigenvalue problem
\begin{equation} \label{eq:EVP PF}
    \gram[XY] \ts \gram[XX] \ts \widetilde{\varphi}_X \overset{!}{=} \lambda \ts \gram[XX] \ts \gram[XX] \widetilde{\varphi}_X.
\end{equation}
Here, $ \gram[XY] = \gram[YX]^\top $. From this, an eigenfunction approximation can again be obtained by setting $ \varphi = \Phi \widetilde{\varphi}_X $. Assuming that $ \gram[XX] $ is regular, this can be simplified by substituting $ \widehat{\varphi}_X := \gram[XX] \widetilde{\varphi}_X $.

\begin{mdframed}[backgroundcolor=boxback,hidealllines=true]
\begin{textalgorithm}
\label{alg:PF_deterministic}
An eigenfunction of the Perron--Frobenius operator can be approximated as follows:
\begin{enumerate}
\item Select a kernel $ k $ and compute the Gram matrices $ \gram[XX] $ and $ \gram[XY] $.
\item Solve the eigenvalue problem $ \gram[XX]^{-1} \ts \gram[XY] \widehat{\varphi}_X = \lambda \widehat{\varphi}_X $.
\item Set $ \varphi = \Phi \ts \gram[XX]^{-1} \ts \widehat{\varphi}_X $ as the approximation to the original eigenproblem \eqref{eq:operator PF eigenproblem}.
\end{enumerate}
\vspace{1pt}
\end{textalgorithm}
\end{mdframed}

If we are only interested in the eigenfunctions evaluated at the training data points, denoted again by $ \varphi_X $, this leads to $ \varphi_X = \widehat{\varphi}_X $ since the $ \Phi $ evaluated in all points $ x_i $ results in $ \gram[XX] $ as shown above. This is consistent with the algorithm for computing eigenfunctions of the Perron--Frobenius operator proposed in Ref.~\onlinecite{KSM17}. To avoid overfitting, $ \gram[XX]^{-1} $ is typically regularized as before.

\subsection{Extension to stochastic systems}
\label{sec:kEDMD_stochastic}

In practice, the methods introduced above can be applied to stochastic dynamical systems as well, by heuristically treating data pairs $(x_i,y_i)$ that are in fact realizations of a stochastic dynamics as generated by a deterministic system. In that situation, however, a large number of data pairs are required in order to compensate for the variance in the stochastic dynamics. The number $m$ of data points that can be taken into account is however limited due to memory constraints given by the $m\times m$ Gram matrices. To address this, we derive \emph{outcome-averaged} Gram matrices, denoted by $ \agram[YX] $ and $\agram[XY]$, that take into account the stochasticity of the system in a more efficient way.

For the stochastic case, the Koopman operator applied to the kernel takes the form
\begin{equation*}
    \big[\ko k(\cdot, x')\big](x) = \int k(y,x) \ts p_\tau(y~|~x')~d\mu(y).
\end{equation*}
The integral can be approximated by Monte Carlo quadrature
\begin{equation}
\label{eq:KO_MonteCarloApprox}
    \big[\ko k(\cdot, x')\big](x) \approx \frac{1}{M} \sum_{l=1}^M k(y^{(l)}, x'),
\end{equation}
where the $ y^{(l)} $ are endpoints of $ M $ independent realizations of the dynamics all starting in $ x $. 

As in Section~\ref{sec:kEDMD_KO}, we discretize \eqref{eq:Koopman EVP in H} by Monte Carlo quadrature and demand it to hold only in $m$ data points $x_i, i=1,\ldots,m$, to obtain
\begin{equation*}
    \sum_{j=1}^m \big[\mathcal{K}k(\cdot,x_j)\big](x_i) \ts \widetilde{\varphi}(x_j)
        \overset{!}{=} \lambda \sum_{j=1}^m k(x_i,x_j) \widetilde{\varphi}(x_j)~.
\end{equation*}
Inserting the new approximation~\eqref{eq:KO_MonteCarloApprox} of the stochastic Koopman operator leads to
\begin{equation*}
    \sum_{j=1}^m \left[\frac{1}{M}\sum_{l=1}^M k\big(y_i^{(l)},x_j\big)\right] \ts \widetilde{\varphi}(x_j)
        \overset{!}{=} \lambda \sum_{j=1}^m k(x_i,x_j) \ts \widetilde{\varphi}(x_j),
\end{equation*}
where the $ y_{i}^{(l)},l=1,\ldots,M $, are independent realizations of the dynamics starting in $ x_i $. With the $m\times m$ matrix
\begin{equation}
\label{eq:G_YX_averaged}
    \big[\agram[YX]\big]_{ij} := \frac{1}{M} \sum_{l=1}^M k\big(y_{i}^{(l)}, x_j\big),
\end{equation}
called the \emph{averaged Gram matrix}, this can be written in matrix form as
\begin{equation*}
    \agram[YX] \widetilde{\varphi}_X \overset{!}{=} \lambda \ts \gram[XX] \widetilde{\varphi}_X.
\end{equation*}

For the Perron--Frobenius operator, the procedure can be repeated in a similar fashion, which then leads to the matrix eigenproblem
\begin{equation*}
    \agram[XY] \widetilde{\varphi}_X \overset{!}{=} \lambda \ts \gram[XX] \widetilde{\varphi}_X,
\end{equation*}
with the averaged Gram matrix $ \agram[XY] = \agram[YX]^\top $. Thus, we can again use the Algorithms~\ref{alg:KO_deterministic} and \ref{alg:PF_deterministic} defined above, with the only difference that we replace the time-lagged Gram matrices by the corresponding averaged counterparts. If we choose $ M = 1 $, we obtain the standard algorithms as a special case, also for stochastic systems.

\subsection{Approximation from time-series data}
\label{sec:timeseriesdata}

In order to compute the averaged Gram matrices \eqref{eq:G_YX_averaged}, many trajectories of length $\tau$ per test point $ x_i $ must be generated, which is straightforward for simple systems, but might be cumbersome or infeasible for complex molecular systems where the data is in general only given in form of a long trajectory. In this situation, multiple realization per starting point $ x_i $ are not available.
We thus propose a simple heuristic method to utilize such trajectory data. The practicality of the method compared to ``standard'' kernel EDMD will be evaluated experimentally in Sections \ref{sec:Quadwell} and \ref{sec:AlanineDipeptide}.

Let---in addition to the test points $\{x_1,\dots,x_m\}$---a much larger data set $\{\tilde{x}_1, \dots,\tilde{x}_M\}$ and its time-$\tau$-evolution $\{\tilde{y}_1,\ldots,\tilde{y}_R\}$ be given. Typically, these data sets come from a long equilibrated trajectory. Define the \emph{trajectory-averaged} Gram matrix $\tagram[YX]$ by
\begin{equation}
\label{eq:G_XY_trajaveraged}
    \big[\tagram[YX]\big]_{ij}
        := \sum_{l=0}^M k(\tilde{y}_l, x_j) \cdot w(\tilde{x}_l,x_i),
\end{equation}
with a certain \emph{weight function} $w(\cdot,\cdot)$. That way, not only trajectories that start exactly at $x_j$, i.e., data pairs $(x_i,y_i)$, contribute to the $ij$th entry of the Gram matrix, but also data pairs $(\tilde{x}_l,\tilde{y}_l)$ where $\tilde{x}_l$ lies near $x_i$. What constitutes ``near'' is defined by the weight function $w$, which should decrease with increasing distance $\|x_i-\tilde{x}_l\|$ and take its maximum for $\tilde{x}_l = x_i$.
Throughout the remainder of the paper, we will use the Gaussian weight function
$$
	w(\tilde{x}_l,x_i) = \frac{1}{Z_i} \exp\left(-\frac{\|x_i-\tilde{x}_l\|^2}{\varepsilon} \right),
$$
where $Z_i = \sum_{k=0}^R w(\tilde{x}_k, x_i)$ is a normalization factor. The bandwidth $\varepsilon$---similar to the parameter $\sigma$ for the Gaussian kernel in Example~\ref{ex:kernels}---controls the influence of surrounding points and has to be chosen problem-dependently. Too large values result in a high contribution from dynamically unrelated points far away, whereas too small values may under-utilize valuable dynamical information from points close by, decreasing the efficiency of the method.

Also note that this rather ad-hoc approach requires the computation of the pairwise distances between the test points $x_i$ and trajectory points $\tilde{x}_l$, which may substantially increase the computational requirement compared to the standard algorithm.
The usage of a cutoff radius or a locally-supported weight functions would help mitigate this.

\subsection{Additional results}

For the derivations above, we assumed that the eigenfunctions are elements of the RKHS $ \mathbb{H} $, which is in general not the case. For a polynomial kernel, for instance, the assumption would imply that all eigenfunctions can be written as polynomials, which is clearly not valid and introduces an approximation error. In Appendix~\ref{app:Kernel transfer operators}, we introduce the notion of \emph{kernel transfer operators}, which can be regarded as approximations of transfer operators in the respective RKHS. Existing methods such as EDMD and VAC and their kernel-based counterparts can be easily interpreted as discretizations of kernel transfer operators. This is described in Appendices~\ref{app:EDMD and VAC} and \ref{app:Kernel EDMD and kernel TICA}, respectively. Furthermore, kernel transfer operators allow for a detailed error analysis with respect to the choice of the kernel, which, however, is beyond the scope of this work.

\section{Applications}
\label{sec:Applications}

We will now demonstrate the application of the introduced algorithms to molecular dynamics problems. Similar examples, analyzed with the aid of related data-driven methods, can be found in Refs.~\onlinecite{SP13, NKPMN14, SP15, KNKWKSN18, Noe_VAMPnets2018}, for instance.

We will refer to the algorithms derived in Sections~\ref{sec:kEDMD_KO} and \ref{sec:kEDMD_PF} as \emph{kernel EDMD for the Koopman-} and \emph{Perron-Frobenius operator}, respectively, the stochastic extension proposed in Section \ref{sec:kEDMD_stochastic} as \emph{stochastic kernel EDMD} and the trajectory-based extension from Section \ref{sec:timeseriesdata} as \emph{trajectory-averaged kernel EDMD}.

\subsection{Quadruple-well problem}
\label{sec:Quadwell}

Here we re-visit the system from Example~\ref{ex:4well_MSMexample}. We want to compute the eigenfunctions of the Perron--Frobenius operator associated with this system using the various methods derived in the previous section.

\subsubsection*{Uniform distribution of test points}

First, we choose a $ 50 \times 50 $ box discretization of $ \mathbb{X} $ and use the centers of the boxes as our test points. That is, $ X $ is sampled from a uniform distribution which corresponds to computing eigenfunctions of the Perron--Frobenius operator $ \pf $ with respect to the Lebesgue measure. As a side remark, this is equivalent to computing eigenfunctions of the Koopman operator $\ko$ with respect to the Gibbs measure, so the results in this section should be compared to the eigenfunctions in Figure~\ref{fig:4well_MSMexample}~(d).
We select the lag time $ \tau = 10 $ and choose the Gaussian kernel $ k(x, y) = \exp\left(- \frac{\norm{x-y}^2}{\sigma}\right) $ with $ \sigma = 0.1 $ and regularization parameter $ \eta = 0.05 $ (see below). 

We first apply the standard Algorithm \ref{alg:PF_deterministic}, i.e., assemble the Gram matrix $ \gram[XY]$ using only one evolved point $y_i=\Theta(x_i)$ for each test point $x_i$.
The resulting eigenvalues and normalized eigenfunctions (computed at all grid points and interpolated in between) are shown in Figure~\ref{fig:Quadruple well PF} (a). The eigenfunctions clearly reveal the expected four metastable sets. Between the fourth and fifth eigenvalue, there is a spectral gap.

\begin{figure*}[htb]
    \centering
    \includegraphics[width=\textwidth]{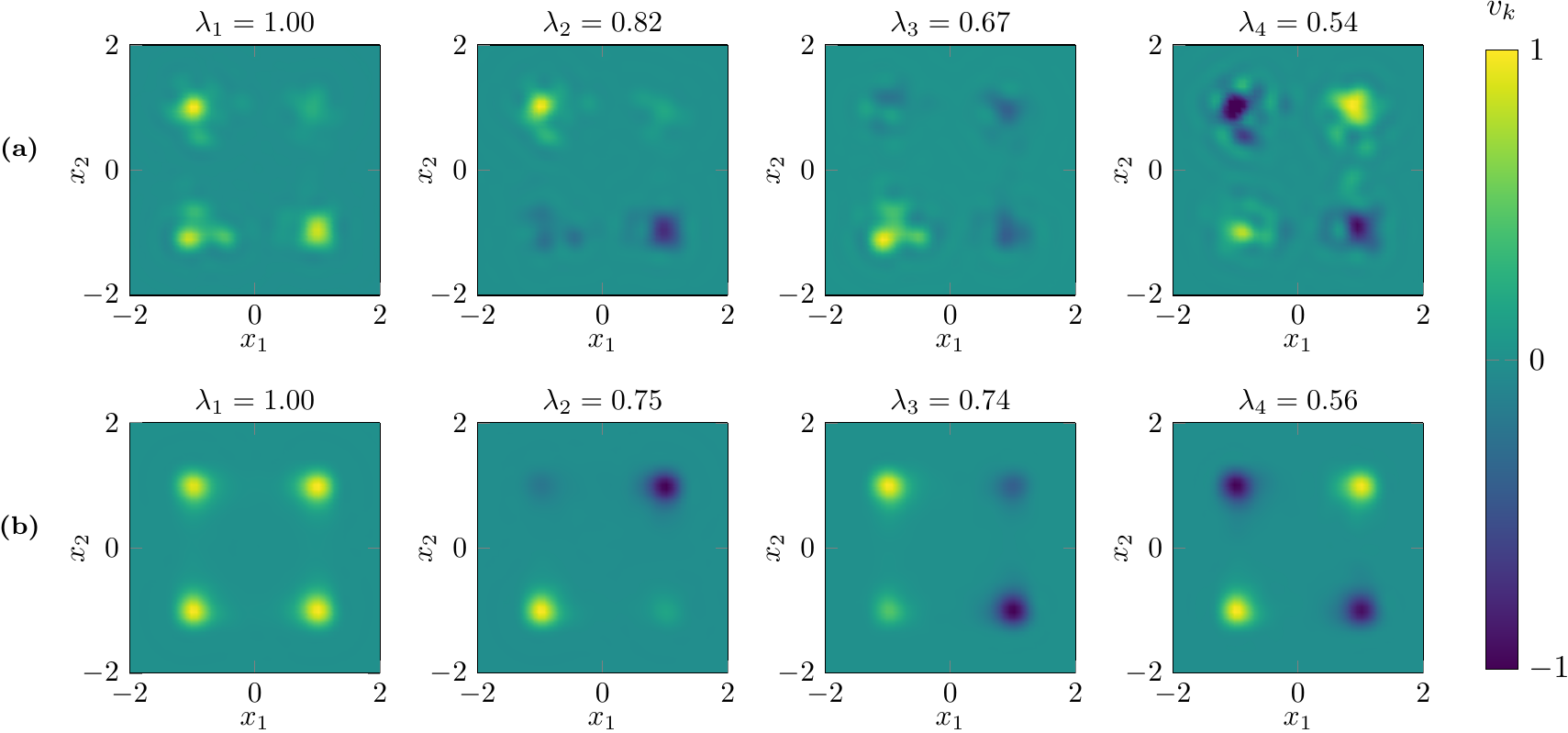}
    \caption{(a) First four eigenfunctions associated with the quadruple-well problem computed using Algorithm \ref{alg:PF_deterministic} where the data points $ x_i $ lie on a regular grid, i.e., are uniformly distributed in $[-2,2]\times[-2,2]$. (b) The same four eigenfunctions, only using the averaged Gram matrix $\agram[XY]$ in Algorithm~\ref{alg:PF_deterministic} with $M=100$ realizations.}
    \label{fig:Quadruple well PF}
\end{figure*}

However, recalling that the analytic eigenfunction $\varphi_1$ is supposed to be identical to the invariant density $\rho\sim e^{-\beta V}$, it is clear that the approximation quality is quite low. We thus replace the standard Gram matrix $ \gram[XY] $ by the averaged Gram matrix $ \agram[XY] $ with $M=100$ realizations of the dynamics per test point. The resulting eigenfunctions can be seen in Figure~\ref{fig:Quadruple well PF} (b).
The invariant density $\rho$ is visually indistinguishable from $\varphi_1$. The dependency of the $L^1$-error, $\|\rho - \varphi_1\|_{L^1}$, on $M$ is shown in Figure~\ref{fig:Quadruple well error M}. The $\mathcal{O}(\frac{1}{\sqrt{M}})$-descend is explained by the better Monte-Carlo approximation of the expectation value \eqref{eq:KO_MonteCarloApprox}.

\begin{figure}[htb]
    \centering
    \begin{minipage}{0.4\textwidth}
        \centering
        \includegraphics[height=0.165\textheight]{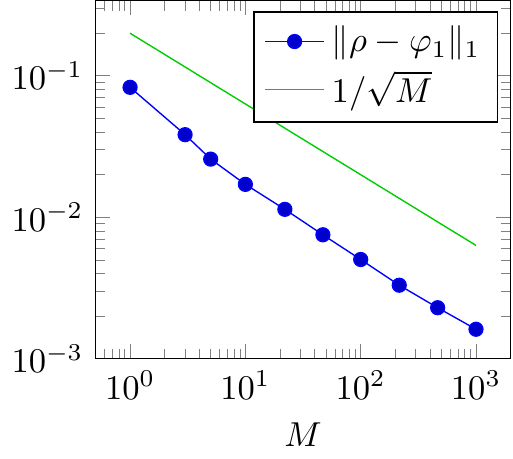}
    \end{minipage}
    \caption{Influence of the number of realizations $M$ on the average deviation between analytically computed invariant density and the dominant Koopman eigenfunction, computed with the stochastic kernel EDMD method for uniformly distributed test points.}
    \label{fig:Quadruple well error M}
\end{figure}

\subsubsection*{Test points from one long trajectory}

In the same way, the eigenfunctions of the Perron--Frobenius operator with respect to the invariant density can be estimated using Algorithm \ref{alg:PF_deterministic}. To this end, we compute one long trajectory and downsample the data to obtain the same lag time as before resulting in a data set containing again 2500 data points. The results are shown in Figure~\ref{fig:Quadruple well TO}~(a). The difference is that the eigenfunctions are not weighted by the equilibrium density anymore as it was the case in Figure~\ref{fig:Quadruple well PF}. Now, the eigenfunctions are almost constant within the four quadrants with sharp transitions between the positive and negative regions. Note that the eigenfunctions of the Perron--Frobenius operator with respect to the Gibbs measure and the eigenfunctions of the Koopman operator with respect to the Lebesgue measure are identical, so Figure~\ref{fig:Quadruple well TO} should be compared to Figure~\ref{fig:4well_MSMexample}~(c).

\begin{figure*}[htb]
    \centering
    \includegraphics[width=\textwidth]{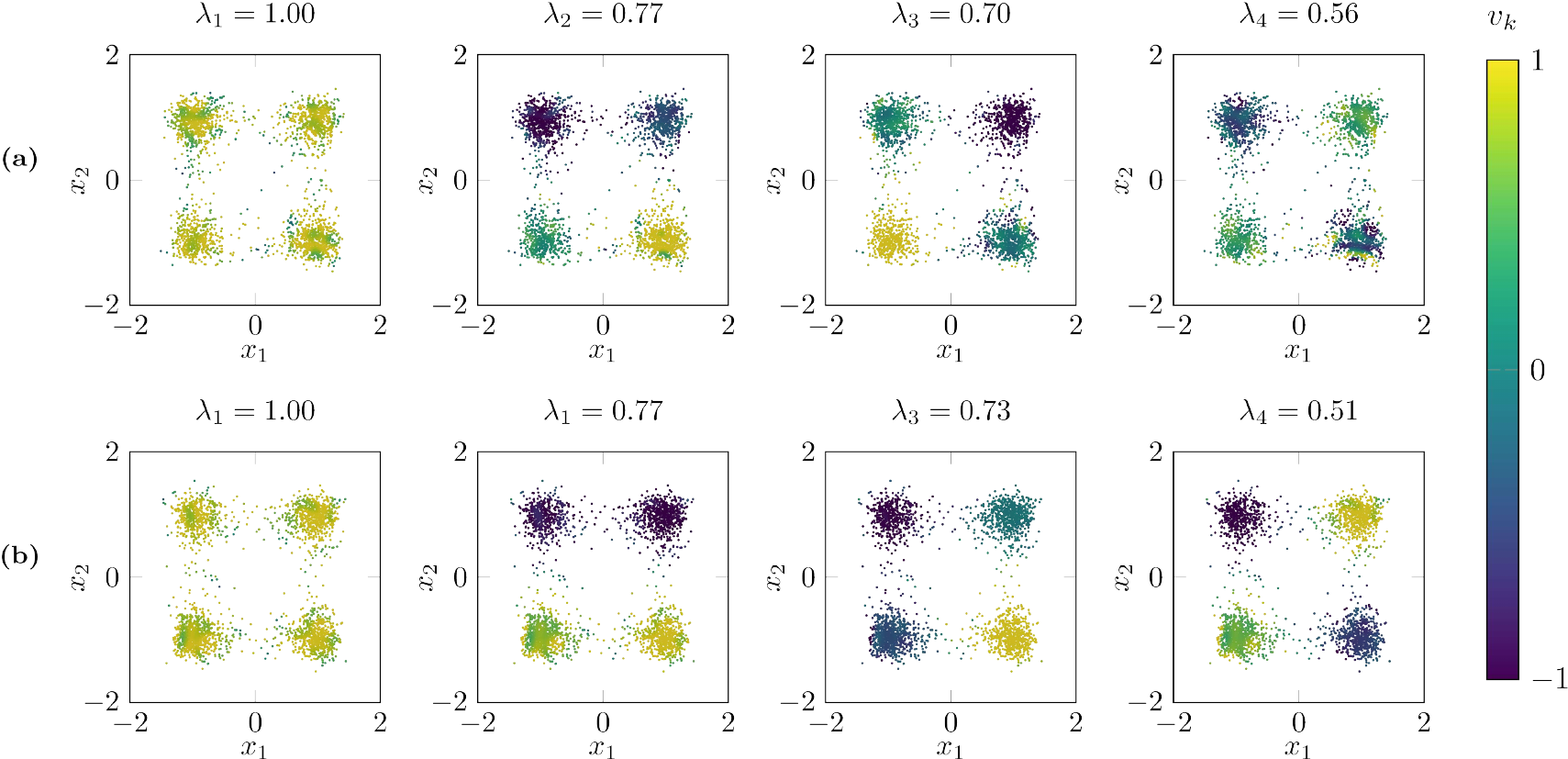}
    \caption{(a) First four eigenfunctions associated with the quadruple-well problem computed using Algorithm \ref{alg:PF_deterministic} where the data points $ x_i $ were extracted from one long trajectory. Thus, the eigenfunctions approximate those of the operator $ \mathcal{P}_\tau $ with respect to the Gibbs measure. (b) The same four eigenfunctions, only using the averaged Gram matrix $\agram[XY]$ in Algorithm~\ref{alg:PF_deterministic} with $M=100$ realizations per test point.}
    \label{fig:Quadruple well TO}
\end{figure*}

While the approximation quality appears to be adequate (the metastable sets are clearly distinguishable by the sign structure of the $\varphi_i$), we observe some noise in the supposedly constant quadrants. Hence, we attempt to reduce this noise by using more dynamic information about the stochastic process instead of just the 2500 data pairs $(x_i,y_i)$. We thus compute the trajectory-averaged Gram matrix \eqref{eq:G_XY_trajaveraged} using a long trajectory of $2.5\cdot 10^5$ frames. The distance parameter in the weight function was chosen as $\varepsilon=0.1$. The result is shown in Figure~\ref{fig:Quadruple well TO}~(b). We observe a small improvement of the results. The dependency of the $L_1$-deviation of the first eigenfunction from the constant function on the trajectory length is shown in in Figure \ref{fig:Quadruple well trajerror}. The stagnation of the error can be delayed by increasing the number of test points.

\begin{figure}[htb]
    \centering
    \centering
    \begin{minipage}{0.4\textwidth}
        \centering
        \includegraphics[height=0.165\textheight]{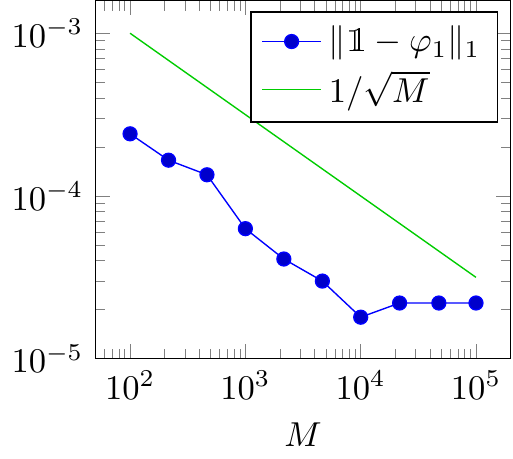}
    \end{minipage}
    \caption{Influence of the trajectory length $M$ on the average deviation between the constant function and the dominant Koopman eigenfunction, computed with the trajectory-averaged kernel EDMD method for Gibbs-distributed test points.}
    \label{fig:Quadruple well trajerror}
\end{figure}

\subsection{Alanine dipeptide}
\label{sec:AlanineDipeptide}

We now use kernel EDMD to analyze the global dynamics of the alanine dipeptide, see Figure~\ref{fig:Alanine dipeptide} (a). It is well known that the global conformational shape of the peptide can be described by the values of only two dihedral angles located in the central backbone \cite{}. A Ramachandran plot, Figure~\ref{fig:Alanine dipeptide} (b), identifies four combinations of these angles that are metastable.

\begin{figure}
	\centering
	\includegraphics[width=\linewidth]{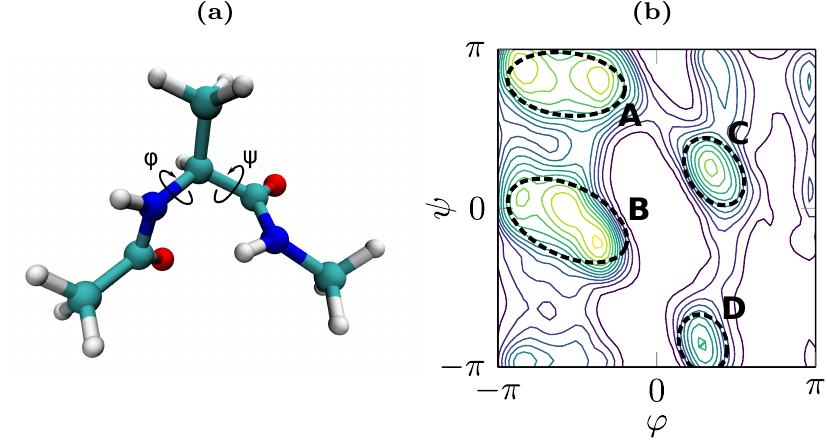}
    \caption{Alanine dipeptide. (a) Structure of the molecule with its two essential dihedral angles $\varphi$ and $\psi$. (b) Ramachandran plot of $(\varphi,\psi)$, revealing four metastable sets.}
    \label{fig:Alanine dipeptide}
\end{figure}

The molecule consists of $22$ atoms, including hydrogen. A global analysis of the corresponding $66$-dimensional systems by conventional grid-based methods such as EDMD is thus out of the question. Instead, we apply the kernel EDMD method for the transfer operator to make the effort conditional on the number of snapshots instead of the number of basis functions. The number of test points was chosen to be $m=4000$, with starting points $x_i$ chosen from a 40 ns long trajectory that was generated with the Gromacs molecular dynamics software\cite{Gromacs}. Thus, the sampling measure $\mu$ is the Gibbs measure. The lag time was defined to be $\tau=20$ ps. In order to make best use of the remaining trajectory data, we employ the trajectory-averaged kernel EDMD method detailed in Section \ref{sec:timeseriesdata} on a $2\cdot 10^5$ frame subsample of the trajectory to assemble the Gram matrix.

We again use the Gaussian kernel. An appropriate choice of the parameter values of $\sigma$ (kernel bandwidth) and $\varepsilon$ (distance parameter in trajectory averaging) is critical for the performance of the method.
While there have been investigations on the optimal choice of the kernel bandwidth\cite{Singer2006}, we are unaware of good strategies for choosing $\varepsilon$ a priori. We thus experimentally examine the influence of the two parameters in more detail. To this end, we post-process for different $(\sigma,\varepsilon)$ combinations the computed eigenvectors with the PCCA+-algorithm\cite{DeuWe05,Roeblitz2013,Web18} in order to find four maximally metastable \emph{membership functions}, i.e., functions $\chi_i:\mathbb{X}\rightarrow [0,1],~i=1,\ldots,4$, that assign each test point $x_i$ its probability to belong to one of four metastable states. This is also called a \emph{fuzzy clustering} of the test points.
For an in-depth introduction to PCCA+ and the relation between membership functions and transfer operator eigenfunctions, we refer to Ref.~\onlinecite{Roeblitz2013}.

Good membership functions are ``unambiguous'', that is, each membership function assigns each test point a membership value close to either zero or one. The objective function
$$
\mathcal{I}[\chi] := \sum_{i=1}^4 \sum_{l=1}^m \chi_i(x_l)^2,
$$
in what follows called the \emph{crispness} of the fuzzy clustering, reflects that and can be thus used as a quality metric of the membership functions and their originating eigenfunctions.

\begin{figure}[htb]
\includegraphics[scale=.7]{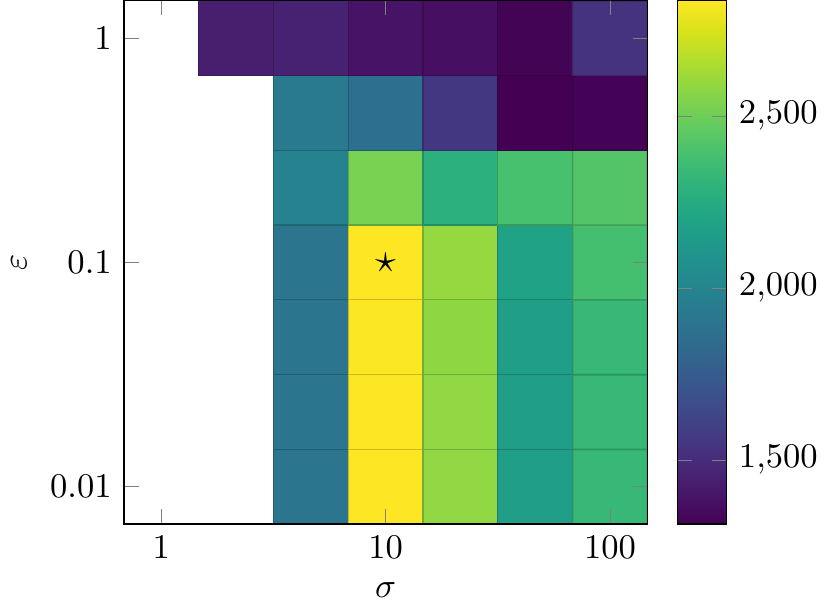}
\caption{Overall crispness of the PCCA+ membership functions for different combinations of $\sigma$ and $\varepsilon$. High values indicate a better quality of the computed eigenvectors. Combinations for which kernel EDMD failed to produce purely real eigenvectors  (a requirement for PCCA+) are left white. The symbol $\star$ marks the optimum parameter combination.}
\label{fig:AD_parametersweep}
\end{figure}

The value of $\mathcal{I}$ for different combinations of $\sigma$ and $\varepsilon$ is shown in Figure \ref{fig:AD_parametersweep}. No regularization has been used to avoid too many parameter dependencies. We observe a maximum value of $\mathcal{I}=2835$ at $\sigma=10$, $\varepsilon=0.1$. A spectral gap is present after the fourth eigenvalue. The associated four membership functions projected onto the $(\varphi,\psi)$-space are shown in Figure~\ref{fig:AD_membershipfunctions} (a) and clearly indicate the four expected metastable sets. A discrete (``hard'') clustering can be constructed by assigning each test point the index of the membership function of maximal value and is shown in Figure \ref{fig:AD_membershipfunctions} (b). These four clusters would then form the four states of an MSM. 

As a remark, the amount of noise in both the fuzzy and hard clustering can be significantly reduced by applying just a slight regularization, e.g., $\eta=0.01$, to the eigenproblem. Assembling the trajectory-averaged Gram matrix $\gram[XY]$ took $596$~s on an eight-core desktop workstation, with the computation of the distances between the test points and trajectory points being the main contributor.

In comparison, the standard kernel EDMD algorithm without trajectory averaging, i.e., using only the test points $x_i$ and single realizations $y_i$ to assemble the Gram matrix, did not result in interpretable eigen- and membership functions when using the same kernel parameter $\sigma=10$ and no regularization. However, increasing the kernel parameter to $\sigma=20$ and using regularization with parameter $\eta=1$ again results in a spectral gap after $\lambda_4$ and membership functions that clearly indicate the metastable sets (Figure \ref{fig:AD_membershipfunctions} (c)). Although the crispness value is lower than for the trajectory--averaged case ($I=2421$), the discrete clustering is less noisy, as the broader kernel and regularization have a smoothing effect on the eigenvectors. Moreover, assembling the standard Gram matrix $\gram[XY]$ is significantly faster, taking only $0.3$~s.

In conclusion, also for realistic MD problems, the heuristic trajectory averaging method results in a measurable improvement of the (fuzzy) spectral clustering if the associated influence parameter $\varepsilon$ can be chosen appropriately, in addition to the kernel bandwidth $\sigma$.
However, if one is interested only in a discrete clustering of the states, for instance to subsequently build an MSM, the standard kernel EDMD method with increased kernel and smoothing parameter yields comparable results at significantly reduced computing cost.

\begin{figure*}
    \centering
    \includegraphics[width=\textwidth]{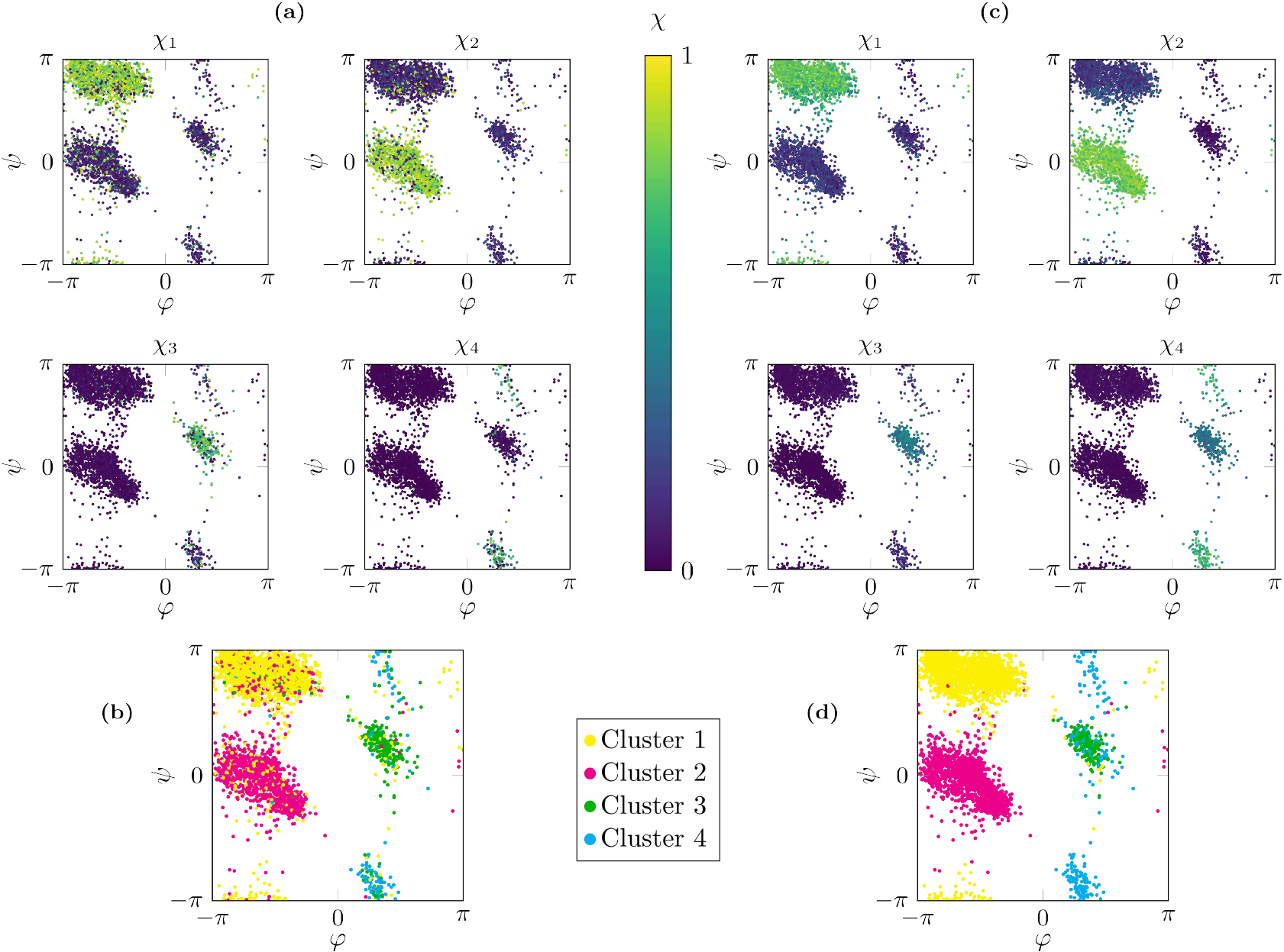}
    \caption{Conformational analysis of the alanine dipeptide. (a) \& (c) PCCA+ membership functions computed from the kernel transfer operator eigenfunctions using the trajectory averaged Gram matrix and the standard Gram matrix, respectively. This represents a ``fuzzy'' clustering of the test points into metastable sets. (b) \& (d) Associated ``hard'' clustering, defined by assigning each test point the index of the membership function of maximum value. }
\label{fig:AD_membershipfunctions}

\end{figure*}

\subsection{NTL9}
\label{sec:NTL9}

Finally, we apply (standard) kernel EDMD to the fast-folding protein NTL9. The molecule consists of 301 heavy atoms (624 atoms overall), divided into 40 residues. As the source of dynamical data we utilize an all-atom simulation of total length 1.11 ms that was performed on the Anton supercomputer \cite{Shaw_QuickFolding2011}.
In order to best capture the system's internal dynamics, we use a Gaussian kernel based on the contact map distance between snapshots, i.e., the kernel
\begin{equation}
\label{eq:NTL9kernel}
k(x,x') = \exp\big(-\|\xi(x) - \xi(y)\|^2/\sigma\big),
\end{equation}
where $\xi(x)$ is the contact map of $x$. The contact map is the $40\times 40$ matrix $A$, whose entries $A_{ij}$ are given by the distances between backbone residue~$i$ and backbone residue~$j$ (i.e., the minimum distance between the atoms in the respective residues). Thus, using this kernel is equivalent to converting the data into a $40\cdot 40=1600$-dimensional state space, which cannot be discretized by traditional grid-based Galerkin methods.

To obtain the data matrix $X$, the trajectory is subsampled with a time step of  100 ns. Likewise, the matrix $Y$ is subsampled from the trajectory, only with an offset of 50 ns from $X$. The rows of $Y$ can thus be interpreted as endpoints of single realizations of length $\tau=\text{50 ns}$ with the starting points being the rows of $X$. The contact maps are computed using the Gromacs software\cite{Gromacs}.

We assemble the Gram matrices $ \gram[XX]$ and $\gram[XY]$ for different values of $\sigma$ and solve the regularized eigenvalue problem \eqref{eq:regularized eigenproblem} with regularization parameter $\eta=0.2$. Figure \ref{fig:NTL9eigenvalues} (a) shows the 15 dominant eigenvalues for varying values of $\sigma$. According to the variational approach of conformational dynamics, higher eigenvalues are in general associated with a better approximation of the original eigenproblem, as discretizations always underestimate the original eigenvalues~\cite{NoNu13}. Thus, the mean value of the dominant eigenvalue can be used as an alternative objective function for choosing the optimal kernel bandwidth. For NTL9, the maximum is observed at $\sigma=10$. Thus, we choose this bandwidth for the further analysis. 

\begin{figure}
    \centering
    \includegraphics[width=\linewidth]{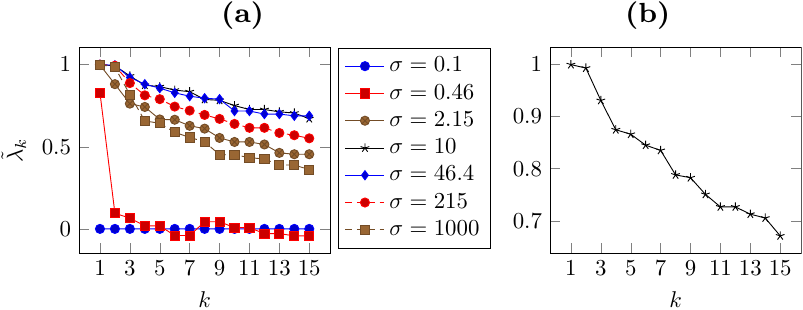}
    \caption{(a) The 14 dominant kernel Koopman operator eigenvalues for different kernel bandwidths $\sigma$. The largest eigenvalues, thus maximal metastability, are observed for $\sigma=10$. (b) Detailed view of the eigenvalues at $\sigma=10$, indicating a spectral gap after $\lambda_1$.}
\label{fig:NTL9eigenvalues}
\end{figure}

A spectral gap can be observed after the second eigenvalue $\tilde{\lambda}_i$ (Figure~\ref{fig:NTL9eigenvalues} (b)), indicating that the data set can be divided into two metastable sets. 
Indeed, clustering the two dominant eigenvectors into two sets using the PCCA+ algorithm as described in the previous section leads to the two conformations shown in Figure \ref{fig:NTL9_clustering}~(a). The conformations coincide very well with the \emph{folded} and \emph{unfolded} states identified in Ref.~\onlinecite{Noe_VAMPnets2018}, where the same data set was analyzed using deep learning methods.
We thus investigate whether we can also reproduce the second hierarchy of metastable conformations identified in Ref.~\onlinecite{Noe_VAMPnets2018} by analyzing the remaining computed eigenfunctions. The result of clustering the eigenvectors into five states using the PCCA+ algorithm can be seen in Figure \ref{fig:NTL9_clustering} (b). Four of the five conformations identified in Ref.~\onlinecite{Noe_VAMPnets2018} could be accurately reproduced by the kernel EDMD algorithm together with PCCA+. The absence of the fifth conformation (named ``Intermediate'' in Ref.~\onlinecite{Noe_VAMPnets2018}) can be explained by the fact that our method uses a factor of ten less data points and that the missing conformation is small (0.6\% of overall states), thus easy to undersample. 

\begin{remark}
The undersampling of high-energy and thus low-population metastable states is a general problem when only a relatively low number of test points is subsampled from Gibbs-distributed trajectory data. 
It has recently been demonstrated in Ref.~\onlinecite{BiBaSch18} how more evenly distributed test points can be generated. 
However, as the distribution of the test points directly influences the form of the transfer operator, and only the eigenpairs of transfer operators associated with the Lebesgue- or the Gibbs measure possess physical interpretations in the context of metastability, a reweighting, for example by importance sampling techniques, has to be introduced to preserve the correct statistics.
\end{remark}

\begin{figure*}
    \centering
    \includegraphics[width=\linewidth]{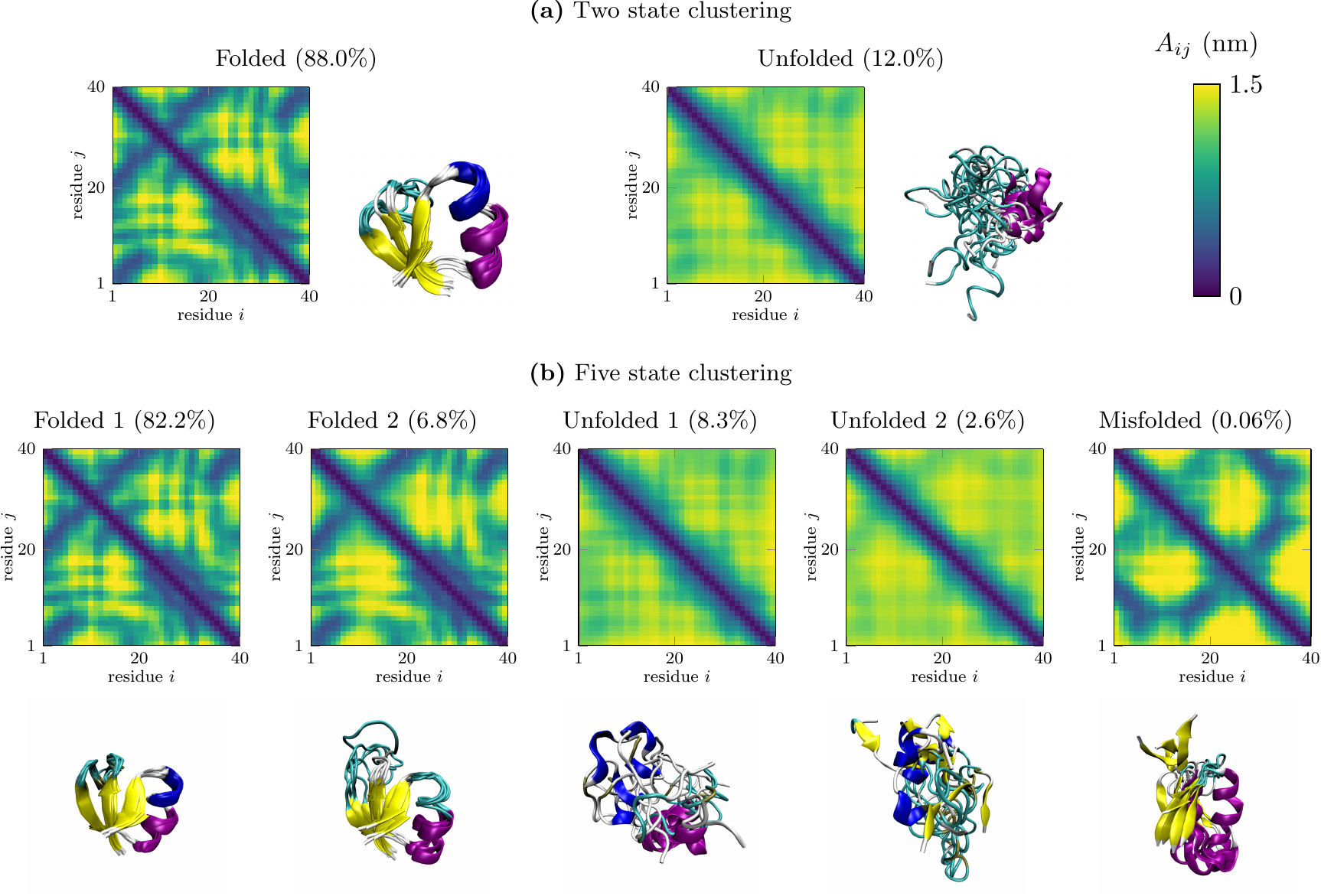}
    \caption{Two state (a) and five state (b) PCCA+ clustering of the dominant NTL9 eigenvectors. The contact maps shown are the mean contact maps of all the points in the respective clusters. The rendered images are the overlays of five randomly-chosen points from the respective cluster. The percentage indicates the relative number of test points in the clusters.}
\label{fig:NTL9_clustering}
\end{figure*}

\section{Conclusion}
\label{sec:Conclusion}

We have shown how the kernel-based discretization of transfer operator eigendecomposition problems leads to efficient and purely data-driven methods for the analysis of molecular conformation dynamics. The derived methods are strongly related to other well-known approaches such as kernel EDMD, kernel TICA, and empirical estimates of kernel transfer operators (as well as conditional mean embeddings). Furthermore, we have proposed an extension to stochastic dynamical systems that is based on averaged Gram matrices and shown that these techniques lead to more accurate spectral decompositions than previous kernel-based approaches. The methods have been applied to realistic molecular dynamics simulations and compared with deep learning techniques, illustrating that kernel-based methods are able to tackle high-dimensional problems for which classical discretization methods fail. Moreover, these methods can also be applied to non-standard representation of molecules, such as contact maps or graph representations of proteins. This enables the user to define similarity measures tailored to specific problems by exploiting domain knowledge about the system.

One possibility to extend the approach presented here would be to combine it with deep kernel learning\cite{wilson2016deep} or other machine learning techniques. Both deep learning and RKHS-based methods are very powerful and their combination might help mitigate the curse of dimensionality typically associated with the approximation of transfer operators pertaining to high-dimensional systems. For non-equilibrium systems, a singular value decomposition might be advantageous. The SVD of empirical estimates of kernel transfer operators (and more general empirical RKHS operators) was recently proposed in Ref.~\onlinecite{MSKS18} and might have additional applications such as low-rank approximation of operators or computing pseudoinverses of operators. Moreover, Gaussian processes might be beneficial to include also uncertainties in these algorithms.

\section*{Acknowledgements}

This research has been partially funded by Deutsche Forschungsgemeinschaft (DFG) through grant CRC 1114 (Projects B03 and B06). We would like to thank D.~E.~Shaw for the NTL9 data set.

\bibliographystyle{unsrt}
\bibliography{KernelMD}

\begin{thebibliography}{10}

\bibitem{Schuette1999}
C.~Sch{\"{u}}tte, A.~Fischer, W.~Huisinga, and P.~Deuflhard.
\newblock {A Direct Approach to Conformational Dynamics Based on Hybrid Monte
  Carlo}.
\newblock {\em J. Comput. Phys.}, 151(1):146--168, 1999.

\bibitem{pande2010everything}
V.~S. Pande, K.~Beauchamp, and G.~R. Bowman.
\newblock {Everything you wanted to know about Markov State Models but were
  afraid to ask}.
\newblock {\em Methods}, 52(1):99--105, 2010.

\bibitem{msm_milestoning}
C.~Sch{\"{u}}tte, F.~No{\'{e}}, J.~Lu, M.~Sarich, and E.~Vanden-Eijnden.
\newblock {Markov state models based on milestoning}.
\newblock {\em J. Chem. Phys.}, 134(20):204105, 2011.

\bibitem{CN14}
J.~D. Chodera and F.~No{\'{e}}.
\newblock {Markov state models of biomolecular conformational dynamics}.
\newblock {\em Curr. Opin. Struct. Biol.}, 25:135--144, 2014.

\bibitem{WKR15}
M.~O. Williams, I.~G. Kevrekidis, and C.~W. Rowley.
\newblock A data-driven approximation of the {K}oopman operator: Extending
  dynamic mode decomposition.
\newblock {\em Journal of Nonlinear Science}, 25(6):1307--1346, 2015.

\bibitem{KKS16}
S.~Klus, P.~Koltai, and C.~Sch{\"u}tte.
\newblock On the numerical approximation of the {P}erron--{F}robenius and
  {K}oopman operator.
\newblock {\em Journal of Computational Dynamics}, 3(1):51--79, 2016.

\bibitem{NoNu13}
F.~No{\'e} and F.~N{\"u}ske.
\newblock A variational approach to modeling slow processes in stochastic
  dynamical systems.
\newblock {\em Multiscale Modeling \& Simulation}, 11(2):635--655, 2013.

\bibitem{NKPMN14}
F.~N\"uske, B.~G. Keller, G.~P\'erez-Hern\'andez, A.~S. J.~S. Mey, and
  F.~No\'e.
\newblock Variational approach to molecular kinetics.
\newblock {\em Journal of Chemical Theory and Computation}, 10(4):1739--1752,
  2014.

\bibitem{KNKWKSN18}
S.~Klus, F.~N\"uske, P.~Koltai, H.~Wu, I.~Kevrekidis, Ch. Sch\"utte, and
  F.~No\'e.
\newblock Data-driven model reduction and transfer operator approximation.
\newblock {\em Journal of Nonlinear Science}, 2018.

\bibitem{MFSS16}
K.~Muandet, K.~Fukumizu, B.~Sriperumbudur, and B.~Sch\"olkopf.
\newblock Kernel mean embedding of distributions: A review and beyond.
\newblock {\em Foundations and Trends in Machine Learning}, 10(1--2):1--141,
  2017.

\bibitem{KSM17}
S.~Klus, I.~Schuster, and K.~Muandet.
\newblock Eigendecompositions of transfer operators in reproducing kernel
  {H}ilbert spaces.
\newblock {\em ArXiv e-prints}, 2017.

\bibitem{Noe_VAMPnets2018}
A.~Mardt, L.~Pasquali, H.~Wu, and F.~No{\'e}.
\newblock {VAMPnets} for deep learning of molecular kinetics.
\newblock {\em Nature Communications}, 9, 2018.

\bibitem{WRK15}
M.~O. Williams, C.~W. Rowley, and I.~G. Kevrekidis.
\newblock A kernel-based method for data-driven {K}oopman spectral analysis.
\newblock {\em Journal of Computational Dynamics}, 2(2):247--265, 2015.

\bibitem{SP15}
C.~R. Schwantes and V.~S. Pande.
\newblock Modeling molecular kinetics with {tICA} and the kernel trick.
\newblock {\em Journal of Chemical Theory and Computation}, 11(2):600--608,
  2015.

\bibitem{LaMa94}
A.~Lasota and M.~C. Mackey.
\newblock {\em Chaos, fractals, and noise: Stochastic aspects of dynamics},
  volume~97 of {\em Applied Mathematical Sciences}.
\newblock Springer, 2nd edition, 1994.

\bibitem{SS13}
C.~Sch\"utte and M.~Sarich.
\newblock {\em Metastability and Markov State Models in Molecular Dynamics:
  Modeling, Analysis, Algorithmic Approaches}.
\newblock Number~24 in Courant Lecture Notes. American Mathematical Society,
  2013.

\bibitem{DJ99}
M.~Dellnitz and O.~Junge.
\newblock On the approximation of complicated dynamical behavior.
\newblock {\em SIAM Journal on Numerical Analysis}, 36(2):491--515, 1999.

\bibitem{Schoe01}
B.~Sch{\"o}lkopf and A.~J. Smola.
\newblock {\em Learning with Kernels: Support Vector Machines, Regularization,
  Optimization and Beyond}.
\newblock MIT press, Cambridge, USA, 2001.

\bibitem{StCh08}
I.~Steinwart and A.~Christmann.
\newblock {\em Support Vector Machines}.
\newblock Springer, 2008.

\bibitem{Borgwardt2005}
K.~M. Borgwardt, C.~S. Ong, S.~Sch{\"o}nauer, S.V.N. Vishwanathan, A.~J. Smola,
  and H.-P. Kriegel.
\newblock Protein function prediction via graph kernels.
\newblock {\em Bioinformatics}, 21(suppl\_1):i47--i56, 2005.

\bibitem{Vishwanathan2010}
S.~V.~N. Vishwanathan, N.~N. Schraudolph, R.~Kondor, and K.~M. Borgwardt.
\newblock Graph kernels.
\newblock {\em Journal of Machine Learning Research}, 11(Apr):1201--1242, 2010.

\bibitem{SP13}
C.~R. Schwantes and V.~S. Pande.
\newblock Improvements in {M}arkov {S}tate {M}odel construction reveal many
  non-native interactions in the folding of {NTL9}.
\newblock {\em Journal of Chemical Theory and Computation}, 9:2000--2009, 2013.

\bibitem{Gromacs}
H.J.C. Berendsen, D.~van~der Spoel, and R.~van Drunen.
\newblock Gromacs: A message-passing parallel molecular dynamics
  implementation.
\newblock {\em Computer Physics Communications}, 91(1):43 -- 56, 1995.

\bibitem{Singer2006}
A.~Singer.
\newblock From graph to manifold laplacian: The convergence rate.
\newblock {\em Applied and Computational Harmonic Analysis}, 21(1):128 -- 134,
  2006.
\newblock Special Issue: Diffusion Maps and Wavelets.

\bibitem{DeuWe05}
P.~Deuflhard and M.~Weber.
\newblock Robust {P}erron cluster analysis in conformation dynamics.
\newblock {\em Linear Algebra and its Applications}, 398:161 -- 184, 2005.
\newblock Special Issue on Matrices and Mathematical Biology.

\bibitem{Roeblitz2013}
S.~R{\"o}blitz and M.~Weber.
\newblock Fuzzy spectral clustering by {PCCA}+: application to markov state
  models and data classification.
\newblock {\em Advances in Data Analysis and Classification}, 7(2):147--179,
  Jun 2013.

\bibitem{Web18}
M.~Weber.
\newblock Implications of {PCCA}+ in molecular simulation.
\newblock {\em Computation}, 6(1), 2018.

\bibitem{Shaw_QuickFolding2011}
K.~Lindorff-Larsen, S.~Piana, R.~O. Dror, and D.~E. Shaw.
\newblock How fast-folding proteins fold.
\newblock {\em Science}, 334(6055):517--520, 2011.

\bibitem{BiBaSch18}
A.~Bittracher, R.~Banisch, and C.~Sch\"utte.
\newblock Data-driven computation of molecular reaction coordinates.
\newblock {\em The Journal of Chemical Physics}, 149(15):154103, 2018.

\bibitem{wilson2016deep}
A.~G. Wilson, Z.~Hu, R.~Salakhutdinov, and E.~P. Xing.
\newblock Deep kernel learning.
\newblock In {\em Artificial Intelligence and Statistics}, pages 370--378,
  2016.

\bibitem{MSKS18}
M.~Mollenhauer, I.~Schuster, S.~Klus, and C.~Sch\"utte.
\newblock Singular value decomposition of operators on reproducing kernel
  {H}ilbert spaces.
\newblock {\em ArXiv e-prints}, 2018.

\bibitem{Baker1973}
C.~Baker.
\newblock Joint measures and cross-covariance operators.
\newblock {\em Transactions of the American Mathematical Society},
  186:273--289, 1973.

\bibitem{Fukumizu04}
K.~Fukumizu, F.~R. Bach, and M.~I. Jordan.
\newblock Dimensionality reduction for supervised learning with {R}eproducing
  {K}ernel {H}ilbert {S}paces.
\newblock {\em Journal of Machine Learning Research}, 5:73--99, 2004.

\bibitem{KoMe17}
M.~Korda and I.~Mezi{\'c}.
\newblock On convergence of {E}xtended {D}ynamic {M}ode {D}ecomposition to the
  {K}oopman operator.
\newblock {\em ArXiv e-prints}, 2017.

\end{thebibliography}

\appendix

\section{Detailed derivations and relationships with other methods}
\label{app:Relationships with other methods}

We will now present an alternative derivation of the operators derived in Section~\ref{sec:Kernel-based discretization of eigenvalue problems} and highlight relationships with other well-known techniques to approximate eigenfunctions of transfer operators, see also Refs.~\onlinecite{KNKWKSN18, KSM17}. To this end, we first introduce kernel transfer operators. Different discretizations of these operators then result in methods such as EDMD or VAC and their kernel-based counterparts. Markov State Models (MSMs), in turn, can be seen as a special case of EDMD where the basis contains indicator functions for a partition of the state space.

\subsection{Mercer feature space}

For the following proofs, we will need the Mercer feature space representation of the kernel, which can be constructed by considering the eigenvalues and eigenfunctions of the integral operator $ \ebd[k] $ associated with the kernel $ k $. The derivation of this representation is mainly based on Ref.~\onlinecite{StCh08}.

There exist an at most countable orthonormal system of eigenfunctions $ (e_\imath)_{\imath \in I} $ and positive eigenvalues $ (\gamma_\imath)_{\imath \in I} \in \ell_1(I) $ of $ \mathcal{E}_k $ such that
\begin{equation*}
    \mathcal{E}_k f = \sum_{\imath \in I} \gamma_\imath \innerprod{f}{e_\imath}_\mu e_\imath
\end{equation*}
for any $ f \in L_2(\mu) $. We assume the eigenvalues $ \gamma_\imath $ to be sorted in non-increasing order. For the sake of simplicity, let $ \phi_\imath = \sqrt{\gamma_\imath} \ts e_\imath $ in what follows.

\begin{theorem}[Mercer's theorem] \label{thm:Mercer}
For $ x, x^\prime \in \inspace $, it holds that
\begin{equation*}
    k(x, x^\prime)
        = \sum_{\imath \in I} \gamma_\imath \ts e_\imath(x) \ts e_\imath(x^\prime)
        = \sum_{\imath \in I} \phi_\imath(x) \ts \phi_\imath(x^\prime),
\end{equation*}
where---for infinite-dimensional feature spaces---the convergence is absolute and uniform \cite{Schoe01}.
\end{theorem}

We can then construct the corresponding RKHS explicitly.

\begin{theorem}
Let $ \mathbb{H} = \left\{ \sum_{\imath \in I} \alpha_\imath \phi_\imath \bigm| (\alpha_\imath)_{\imath \in I} \in \ell_2(I) \right\} $ and $ f, g \in \mathbb{H} $, with $ f = \sum_{\imath \in I} \alpha_\imath \ts \phi_\imath $ and $  g = \sum_{\imath \in I} \beta_\imath \ts \phi_\imath $. Define
\begin{equation*}
    \innerprod{f}{g}_\mathbb{H} = \sum_{\imath \in I} \alpha_\imath \ts \beta_\imath,
\end{equation*}
then $ \mathbb{H} $ with this inner product is the RKHS associated with the kernel $ k $.
\end{theorem}

Note that the kernel defined above indeed has the reproducing property: Given a function $ f = \sum_{\imath \in I} \alpha_\imath \ts \phi_\imath \in \mathbb{H} $, we obtain using the definition of the inner product for~$ \mathbb{H} $
\begin{align*}
    \innerprod{f}{k(x, \cdot)}_\mathbb{H}
        &= \Big\langle \sum_{\imath \in I} \alpha_\imath \phi_\imath ,\, \sum_{\imath \in I} \phi_\imath(x) \ts \phi_\imath\Big\rangle_\mathbb{H} \\
        &= \sum_{\imath \in I} \alpha_\imath \ts \phi_\imath(x) = f(x).
\end{align*}

We define the projection of a function onto the reproducing kernel Hilbert space as follows:

\begin{definition}[Orthogonal projection]
Given a function $ f \in L_2(\mu) $, the \emph{orthogonal projection} onto $ \mathbb{H} $ is defined as
\begin{equation*}
    \pro[k] \ts f
        = \sum_{\imath \in I} \frac{\innerprod{f}{\phi_\imath}_\mu}{\innerprod{\phi_\imath}{\phi_\imath}_\mu} \phi_\imath
        = \sum_{\imath \in I} \innerprod{f}{e_\imath}_\mu e_\imath.
\end{equation*}
\end{definition}

\subsection{Covariance operators}

Let $ \phi $ be the Mercer feature space representation associated with the kernel $ k $. Instead of writing $ f = \sum_{\imath \in I} \alpha_\imath \ts \phi_\imath $, we will also use the shorter form $ f = \alpha^\top \phi $, even though $ I $ might be a (countably) infinite set.

\begin{definition}[Covariance operators\cite{Baker1973, MFSS16}]
Given $ f = \alpha^\top \phi \in \mathbb{H} $. The \emph{covariance operator} $ \cov[XX] \colon \mathbb{H} \to \mathbb{H} $ and the \emph{cross-covariance operator} $ \cov[XY] \colon \mathbb{H} \to \mathbb{H} $ are defined as
\begin{equation*}
    \cov[XX] f = (\mcov[XX] \ts \alpha)^\top \phi
    \quad \text{and} \quad
    \cov[XY] f = (\mcov[XY] \ts \alpha)^\top \phi,
\end{equation*}
where
\begin{equation*}
    \mcov[XX] = \int \phi(x) \otimes \phi(x) \ts \dd \mu(x)
\end{equation*}
and
\begin{equation*}
    \mcov[XY]  = \int \phi(x) \otimes \ko \phi(x) \ts \dd \mu(x)
\end{equation*}
and $ \otimes $ denotes the tensor product.
\end{definition}

For vector-valued functions, the Koopman operator is defined to act componentwise. In Refs.~\onlinecite{MFSS16, KSM17}, the covariance operator is defined with respect to the joint probability measure. However, note that due to the definition of the Koopman operator this is equivalent since $ \ko \phi(x) = \mathbb{E}_{\scriptscriptstyle Y \mid X}[\phi(Y) \mid X = x] $. Given training data $ (x_i, y_i) $, $ i = 1, \dots, m $, as above, we define
\begin{equation*}
    \Phi = [\phi(x_1), \dots, \phi(x_m)] \quad \text{and} \quad \Psi = [\phi(y_1), \dots, \phi(y_m)].
\end{equation*}
The empirical estimates of the matrix representations of the operators are given by
\begin{alignat*}{4}
    \mecov[XX] &= \frac{1}{m} \sum_{i=1}^m \phi(x_i) \otimes \phi(x_i)
               &&= \frac{1}{m} \Phi \ts \Phi^\top, \\
    \mecov[XY] &= \frac{1}{m} \sum_{i=1}^m \phi(x_i)\otimes \phi(y_i)
               &&= \frac{1}{m} \Phi \ts \Psi^\top.
\end{alignat*}

\begin{remark}
Note that for the feature matrix $ \Phi $ in Section~\ref{sec:Kernel-based discretization of eigenvalue problems} we used the canonical feature map $ x \mapsto \phi(x) = k(x, \cdot) $, whereas we now use the Mercer feature map $ x \mapsto \phi(x) = [\phi_1(x), \phi_2(x), \dots]^\top $, where $ \phi(x) $ is an element in a potentially infinite-dimensional vector space. As long as our algorithms rely only on kernel evaluations, it does not matter which feature space we consider since for both $ \innerprod{\phi(x)}{\phi(x^\prime)}_\mathbb{H} = k(x, x^\prime) $ and the resulting Gram matrices $ \gram[XX] = \Phi^\top \Phi $ are identical.
\end{remark}

\subsection{Properties of the integral operator}
\label{app:Properties of the integral operator}

The following lemma summarizes useful properties of the different operators.

\begin{lemma} \label{lem:operator properties}
It holds that:
\begin{enumerate}[label=(\roman*)]
\item $ \mathcal{E}_k f = \cov[XX] \pro[k] f $ for $ f \in L_2(\mu) $.
\item $ \ebd[k] \ts \pro[k] \ts f = \pro[k] \ts \ebd[k] \ts f = \ebd[k] \ts f $ for $ f \in L_2(\mu) $.
\item Any function $ f \in \mathbb{H} $ can be written as $ f = \mathcal{E}_k \tilde{f} $ with $ \tilde{f} \in \mathbb{H} $ provided that $ \left(\frac{\alpha_i}{\gamma_\imath}\right)_{\imath \in I} \in \ell_2(I) $.
\end{enumerate}
\end{lemma}

\begin{proof} The properties follow directly from the series expansion of $ \ebd[k] $. \\[1ex]
\textit{(i)} The matrix $ \mcov[XX] $ is a diagonal matrix whose entries are the eigenvalues $ \gamma_\imath $ of the integral operator $ \ebd[k] $ since
\begin{equation*}
    [\mcov[XX]]_{\imath \jmath}
        = \int \phi_\imath(x) \ts \phi_\jmath(x) \ts \dd \mu(x)
        = \sqrt{\gamma_\imath} \ts \sqrt{\gamma_\jmath} \ts \delta_{\imath\jmath}
\end{equation*}
so that $ \cov[XX] \ts \pro[k] \ts f = \sum_{\imath \in I} \gamma_\imath \ts \innerprod{f}{e_\imath}_\mu \ts e_\imath = \ebd[k] f $. \\[1ex]
\textit{(ii)} $ \pro[k]^2 = \pro[k] $ and $ \cov[XX] $ is a mapping from $ \mathbb{H} $ to $ \mathbb{H} $. \\[1ex]
\textit{(iii)} For $ f = \alpha^\top \phi $, define $ \tilde{f} = (\mcov[XX]^{-1} \ts \alpha)^\top \phi $, where $ \mcov[XX]^{-1} = \diag(\gamma_\imath^{-1})_{\imath \in I} $.
\end{proof}

It follows in particular that $ \mathcal{E}_k f = \cov[XX] f $ for $ f \in \mathbb{H} $.

\begin{remark}
For infinite-dimensional feature spaces, the assumption $ \left(\frac{\alpha_i}{\gamma_\imath}\right)_{\imath \in I} \in \ell_2(I) $ will not be satisfied for all $ f \in \mathbb{H} $. By setting $ \alpha_\imath = 0 $ for $ \imath > N $, where $ N $ is fixed, the equation can, however, be satisfied approximately. In practice, since only finitely many data points are considered, we can always construct a finite-dimensional feature map (e.g., the data-dependent kernel map, see Ref.~\onlinecite{Schoe01}).
\end{remark}

\subsection{Kernel transfer operators}
\label{app:Kernel transfer operators}

Kernel transfer operators, which were introduced in Ref.~\onlinecite{KSM17}, can be regarded as approximations of transfer operators in RKHSs. The underlying assumption is that both the densities or observables and the densities or observables propagated by the Perron--Frobenius or Koopman operator, respectively, are functions in the RKHS $ \mathbb{H} $, which, in general, is not the case. Empirical estimates of these operators result, under certain conditions, in methods such as EDMD \cite{WKR15} or VAC \cite{NKPMN14} or their kernel-based counterparts \cite{WRK15,SP15}. Let us start with a brief description of kernel transfer operators, see also Ref.~\onlinecite{KSM17}. The derivation here is slightly different in that we directly use the Mercer feature space representation. The goal is to illustrate connections with Galerkin approximations.

\begin{proposition} \label{thm:covariance operators}
Let $ f, g \in \mathbb{H} $ with $ f = \alpha^\top \phi $ and $ g = \beta^\top \phi $, then
\begin{align*}
    \innerprod{f}{g}_\mu
        &= \innerprod{f}{\cov[XX] \ts g}_\mathbb{H}
         = \alpha^\top \mcov[XX] \ts \beta, \\
    \innerprod{f}{\ko g}_\mu
        &= \innerprod{f}{\cov[XY] g}_\mathbb{H}
         = \alpha^\top \mcov[XY] \ts \beta.
\end{align*}
\end{proposition}
\begin{proof}
We show the second part, the first one follows analogously:
\begin{align*}
    \innerprod{f}{\ko g}_\mu
        &= \int \alpha^\top \phi(x) \ts \beta^\top \ko \phi(x) \ts \dd \mu(x) \\
        &= \alpha^\top \left[ \int \phi(x) \otimes \ko \phi(x) \ts \dd \mu(x) \right] \beta \\
        &= \alpha^\top \ts \mcov[XY] \ts \beta.
\end{align*}
Similarly,
\begin{align*}
    \innerprod{f}{\cov[XY] g}_\mathbb{H}
        &= \innerprod{\alpha^\top \phi}{\cov[XY] (\beta^\top \phi)}_\mathbb{H} \\
        &= \innerprod{\alpha^\top \phi}{(\mcov[XY] \ts \beta)^\top \phi}_\mathbb{H} \\
        &= \alpha^\top \ts \mcov[XY] \ts \beta,
\end{align*}
where we used the definition of the inner product in $ \mathbb{H} $ in the last step.
\end{proof}

\begin{proposition} \label{thm:Fukumizu04}
Assuming that $ \ko g \in \mathbb{H} $ for all $ g \in \mathbb{H} $, it holds that $ \cov[XX] \ko g = \cov[XY] g $.
\end{proposition}

The proof---formulated in terms of expectation values rather than transfer operators---can be found in Ref.~\onlinecite{Fukumizu04}. The assumption corresponds to $ \mathbb{H} $ being an invariant subspace of the Koopman operator. If $ \cov[XX] $ is invertible, we can immediately define the kernel Koopman operator as $ \ko[k] = \cov[XX]^{-1} \ts \cov[XY] $. Otherwise, the regularized version $ (\cov[XX] + \eta \ts \mathcal{I})^{-1} $ is used. An analogous result can be obtained for the Perron--Frobenius operator.

\begin{proposition}
Assuming that $ \pf g \in \mathbb{H} $ for all $ g \in \mathbb{H} $, we obtain $ \cov[XX] \pf g = \cov[XY] g $.
\end{proposition}
\begin{proof}
Using the duality of the Perron--Frobenius and Koopman operator as well as Propositions~\ref{thm:covariance operators} and \ref{thm:Fukumizu04}, we can write
\begin{align*}
    \innerprod{f}{\cov[XX] \pf g}_\mathbb{H}
        &= \innerprod{f}{\pf g}_\mu
        = \innerprod{\ko f}{g}_\mu \\
        &= \innerprod{\cov[XX] \ko f}{g}_\mathbb{H}
        = \innerprod{\cov[XY] f}{g}_\mathbb{H} \\
        &= \innerprod{f}{\cov[YX] g}_\mathbb{H}. \qedhere
\end{align*}
\end{proof}

A similar result was shown in Ref.~\onlinecite{KSM17}. It follows that we can define the kernel Perron--Frobenius operator as $ \pf[k] = \cov[XX]^{-1} \ts \cov[YX] $, where regularization might be required as above. In general, however, $ \mathbb{H} $ will not be invariant under the action of the Koopman or Perron--Frobenius operator.

\subsection{EDMD and VAC}
\label{app:EDMD and VAC}

In what follows, we will analyze the convergence properties of the empirical estimates of the kernel transfer operators, given by
\begin{equation*}
    \eko[k] = \ecov[XX]^{-1} \ts \ecov[XY] = \big(\Phi \Phi^\top\big)^{-1} \big(\Phi \Psi^\top\big)
\end{equation*}
and
\begin{equation*}
    \epf[k] = \ecov[XX]^{-1} \ts \ecov[YX] = \big(\Phi \Phi^\top\big)^{-1} \big(\Psi \Phi^\top\big).
\end{equation*}
It is important to remark that these operator approximations can only be computed numerically if the feature space of the kernel is finite-dimensional. However, the resulting eigenvalue problems can be reformulated in such a way that only kernel evaluations are required. The feature space representation is then never computed explicitly, but only implicitly through the kernel. Assuming that the feature space is finite-dimensional, we obtain the following algorithm to estimate eigenfunctions of transfer operators, where $ ^H $ denotes the Hermitian transpose.

\begin{mdframed}[backgroundcolor=boxback,hidealllines=true]
\begin{textalgorithm} \label{alg:EDMD}
In order to approximate eigenfunctions of the Koopman operator:
\begin{enumerate}
\item Select a set of basis functions $ \phi $ and compute the matrices $ \ecov[XX] $ and $ \ecov[XY] $.
\item Solve the eigenvalue problem $ \ecov[XX]^{-1} \ts \ecov[XY] v = \lambda \ts v $.
\item Set $ \varphi = v^H \phi $.
\end{enumerate}
\vspace{1pt}
\end{textalgorithm}
\end{mdframed}

For the Perron--Frobenius operator, it suffices to replace $ \cov[XY] $ by $ \cov[YX] $ as described above. These empirical estimates are equivalent to VAC and EDMD, see also Ref.~\onlinecite{KSM17}. That is, in the infinite-data limit, we obtain the Galerkin approximation of the corresponding operator in the RKHS~$ \mathbb{H} $, cf.~Refs.~\onlinecite{WKR15, KKS16}. Note that for EDMD $ \phi $ does not necessarily have to be the Mercer feature space representation of the kernel. The Mercer feature space is orthogonal with respect to the inner product in $ L_2(\mu) $ so that the empirical estimate of the covariance matrix $ \ecov[XX] $ converges to $ \cov[XX] = \diag(\gamma_1, \dots, \gamma_n) $ and its inverse is simply $ \cov[XX]^{-1} = \diag(\gamma_1^{-1}, \dots, \gamma_n^{-1}) $. The convergence of EDMD to the Koopman operator for $ n \to \infty $, where $ n $ is the dimension of the state space, i.e., the cardinality of the index set $ I $, was first studied in Ref.~\onlinecite{KoMe17}.

The equivalence of the above algorithm and the methods derived in Section~\ref{sec:Kernel-based discretization of eigenvalue problems}---from a linear algebra point of view---was shown in Ref.~\onlinecite{WRK15}, where EDMD for the Koopman operator was \emph{kernelized} to obtain kernel EDMD. Similarly, equivalence of the different empirical RKHS operator eigenvalue problems---from a functional analytic point of view---was shown in Ref.~\onlinecite{MSKS18}.

\subsection{Kernel EDMD and kernel TICA}
\label{app:Kernel EDMD and kernel TICA}

In Ref.~\onlinecite{WRK15}, kernel EDMD for the Koopman operator was derived by kernelizing standard EDMD. A different derivation based on kernel transfer operators and embedded kernel transfer operators was presented in Ref.~\onlinecite{KSM17}. For $ M = 1 $, which corresponds to no averaging to compute time-lagged Gram matrices, the methods derived in Section~\ref{sec:Kernel-based discretization of eigenvalue problems} are identical to (extensions of) kernel EDMD. The main difference is that in Section~\ref{sec:Kernel-based discretization of eigenvalue problems} we directly discretized the transfer operator eigenvalue problem. This allows for exploiting the stochasticity of the system to improve the numerical stability of kernel-based approaches for molecular dynamics problems and other stochastic dynamical systems.

Similarly, our approach is also related to kernel TICA, which was proposed in Ref.~\onlinecite{SP15}. However, for the derivation of kernel TICA the system is assumed to be reversible. In order to obtain a real spectrum, the trajectory and time-reversed trajectory are used for the approximation. This can be accomplished by defining feature matrices $ \widehat{\Phi} = [\Phi, \; \Psi] $ and $ \widehat{\Psi} = [\Psi, \; \Phi] $ and applying the kernel-based methods to this new data set, which doubles the size of the resulting eigenvalue problem. Furthermore, their derivation is based on a variational principle and leads to a slightly different problem, which---using our notation---can be written as
\begin{equation*}
    \rgram[XX] \ts R \ts \rgram[XX] \, \beta = \lambda (\rgram[XX] \ts \rgram[XX] + \eta \ts I) \beta,
\end{equation*}
where $ \rgram[XX] $ is the augmented Gram matrix given by
\begin{equation*}
    \rgram[XX] = \widehat{\Phi}^\top \widehat{\Phi} = \begin{bmatrix} \gram[XX] & \gram[XY] \\ \gram[YX] & \gram[YY] \end{bmatrix}
    \quad \text{and} \quad
    R = \begin{bmatrix} 0 & I \\ I & 0 \end{bmatrix}.
\end{equation*}
Note that for the augmented Gram matrices
\begin{equation*}
    \rgram[XX] R =
    \begin{bmatrix}
        \gram[XX] & \gram[XY] \\ \gram[YX] & \gram[YY]
    \end{bmatrix}
    \begin{bmatrix}
        0 & I \\ I & 0
    \end{bmatrix}
    =
    \begin{bmatrix}
        \gram[XY] & \gram[XX] \\ \gram[YY] & \gram[YX]
    \end{bmatrix}
    = \rgram[XY].
\end{equation*}
Thus, this is equivalent to \eqref{eq:EVP PF} for the time-reversed system. As before, for large enough $ \eta $, the regularization term guarantees that the matrix on the right will be positive definite.

\end{document}